\documentclass[11pt,a4paper]{amsart}

\usepackage{fancyhdr} %for design
\usepackage{stmaryrd} %for the \leftslice and \rightslice commands from design
\usepackage{combelow} %for the \cb{} command
\usepackage{hyperref} %for the \texorpdfstring{} command
\usepackage{titlesec} %for the titleformat command in design
\usepackage[a4paper]{geometry} %for design
\usepackage{xstring,xifthen} %for design
\usepackage{xcolor} %for the \textcolor command
\usepackage{subcaption}  % modern replacement for subfigure
\usepackage{graphicx} % for figures
\usepackage{listings} % for code listings
\usepackage{array} % for table columns manipulation
\usepackage{float} %for H placement of figures
\usepackage{tabularx} %for tables
\usepackage{multirow} %for multirow tables
\usepackage{enumitem}

\newtheorem{theorem}{Theorem}[section]

\newtheorem{lemma}[theorem]{Lemma}
\newtheorem{proposition}[theorem]{Proposition}
\theoremstyle{definition}
\newtheorem{definition}[theorem]{Definition}
\newtheorem{algorithm}[theorem]{Algorithm}

\newtheorem{remark}[theorem]{Remark}

\newcommand{\Sun}{\mathcal{S}}
\newcommand{\M}{\mathcal{M}}
\newcommand{\A}{\mathcal{A}}
\newcommand{\U}{\mathcal{U}}
\renewcommand{\d}{\mathrm{d}}
\renewcommand{\P}{\mathcal{P}}

\newcommand{\titlename}	
{Short-term solar eclipse predictions using Euclidean geometry and dynamical systems}
\newcommand{\shorttitlename}
{Short-term solar eclipse predictions using Euclidean geometry and dynamical systems}

\newcommand{\authorname}      {Alin-Doru Aluaș$^1$, Valerian-Alin Fodor$^2$ \and Virgilius-Aurelian Minuță$^3$}
\newcommand{\pdfauthorname}   {Alin-Doru Aluaș, Valerian-Alin Fodor, Virgilius-Aurelian Minuță}
\newcommand{\shortauthorname} {A. D. Aluaș, V. A. Fodor \and V. A. Minuță}

\newcommand{\universitynameA}  {$^{1,2,3}$Babeș-Bolyai University}
\newcommand{\facultynameA}     {Faculty of Mathematics and Computer Science}
\newcommand{\departmentnameA}  {Department of Mathematics}
\newcommand{\addressA}  	   {M. Kogălniceanu 1, RO-400084, Cluj-Napoca, Romania}

\newcommand{\emailaddressA}    {alinaluasdoru@gmail.com}
\newcommand{\emailaddressB}    {valerian.fodor@ubbcluj.ro}
\newcommand{\emailaddressC}    {virgilius.minuta@ubbcluj.ro}

\newcommand{\articleabstract}{We propose a computationally inexpensive algorithm to determine whether an opaque spherical body eclipses a luminous one, using only classical geometrical tools. We apply this method to the Sun-Earth-Moon system and predict solar eclipses, including their type, beginning, and duration, using the three-body problem. For predictions over a one-year period, the resulting errors are in the order of minutes.}

\newcommand{\msc}{51P05, 85A04, 70F07.}
\newcommand{\keywordterms}{Eclipse predictions, Three-body problem, Solar eclipses, Cone-sphere intersection.}

\def\depA{\departmentnameA}
\StrLen{\depA}[\depAlen]
\def\facA{\facultynameA}
\StrLen{\facA}[\facAlen]

\newcommand{\institutionA}{
\universitynameA\\
\ifthenelse{\facAlen>0}{\facultynameA\\}{}
\ifthenelse{\depAlen>0}{\departmentnameA\\}{}
\addressA}
\hypersetup{colorlinks = true, linkcolor = blue, anchorcolor =red, citecolor = blue, filecolor = red, urlcolor = blue, pdfauthor=\pdfauthorname, pdftitle=\titlename, pdfkeywords=\keywordterms, pdfsubject=\articleabstract}
\titleformat{\section}{\Large\bfseries}{\thesection}{1em}{}
\titleformat{\subsection}{\large\it}{\thesubsection}{1em}{}
\fancypagestyle{mypagestyle}{
  \fancyhf{}% Clear header/footer
  \fancyhead[OC]{\textit{\shorttitlename}}
  \fancyhead[EC]{\textit{\shortauthorname}}
  \fancyfoot[C]{\medskip $\leftslice\,$\thepage$\,\rightslice$ }
  
}
\fancypagestyle{firstpagestyle}{
  \fancyhf{}% Clear header/footer
  \fancyfoot[C]{\medskip $\leftslice\,$\thepage$\,\rightslice$ }
  
}
\title[\shorttitlename]{\LARGE{\titlename}}
\author[\shortauthorname]{\large{\authorname}
\medskip\\
{\footnotesize \institutionA\medskip\\
$\begin{array}{l}
\text{email}^1\text{: \texttt{\href{mailto:\emailaddressA}{\emailaddressA}}}\\
\text{email}^2\text{: \texttt{\href{mailto:\emailaddressB}{\emailaddressB}}}\\
\text{email}^3\text{: \texttt{\href{mailto:\emailaddressC}{\emailaddressC}}}
\end{array}$
}}
\begin{document}
\begin{abstract}
\articleabstract\\[0.1cm]
\textsc{MSC 2010.} \msc\\[0.1cm]
\end{abstract}
\begingroup
\def\uppercasenonmath#1{} %disables uppercase of title
\let\MakeUppercase\relax %disables uppercase of authors
\maketitle
\endgroup
\thispagestyle{firstpagestyle}

%%%%%%%%%%%%%%%%%%%%%%%%%%%%%%%%%%%%%%%%%%%%%%%%%%%%%%%%%%%%%%%
%%%%%%%%%%%%%%%%%%%%Start of document%%%%%%%%%%%%%%%%%%%%%%%%%%
%%%%%%%%%%%%%%%%%%%%%%%%%%%%%%%%%%%%%%%%%%%%%%%%%%%%%%%%%%%%%%%
\section{Introduction} \label{sec:introduction}

In order to detect whether an opaque spherical body eclipses a luminous one, modern methods involve the use of general relativity and complex celestial mechanics. In this article, we investigate whether a simple classical geometry approach would suffice for the prediction of eclipses on a short time-scale.

In order to realize that, at the core of our paper, we analyze the intersection between a cone and a sphere (Section \ref{subsec:double_cone}), by providing two alternative methods of determining whether a double cone intersects a given sphere. The former method is based on the characterization of the double cone (Theorem \ref{th:double_cone}) and requires solving a Lagrangian system \eqref{eq:Lagrange_system}, while the latter (Algorithm \ref{algo:intersection_sphere_double_cone}) is of a geometrical nature with less expensive computations, henceforth all following algorithms (for the intersection between a sphere and a convex circular cone, a circular cone and a frustum of a convex circular cone, respectively) will have Algorithm \ref{algo:intersection_sphere_double_cone} at their root. 

Section \ref{sec:geometry} continues by delving into the geometry of eclipses and establishes formulas for obtaining the parameters of the so-called eclipse cones (Remark \ref{remark:eclipse_parameters} and Remark \ref{remark:total_vs_annular_eclipse}).

While Section \ref{sec:geometry} considers a static picture of celestial bodies (i.e. at the root of all formulas and algorithms of Section \ref{sec:geometry} lies the position of said  celestial bodies), Section \ref{sec:3bps} also takes into account their motion via the three-body problem. Hence, taking as a particular case study, the Sun-Earth-Moon system, with initial parameters obtained from the \emph{NASA Jet Propulsion Laboratory (JPL) HORIZONS On-Line Ephemeris System} \cite{JPL:HORIZONS} at an initial epoch of $$t_0 =\text{ 1 January 2026, 00:00 UTC;}$$ and by numerically approximating the solutions of this three-body problem using the computer algebra system Maple, we were able to simulate the trajectories of our celestial bodies in question, for a period of one year with a low-error (less then $0.00065$ AU in drift, at the end of the simulation period, compared with the more complex simulated data from NASA for the same period).

\goodbreak
Finally, in Section \ref{sec:results}, we give our eclipse detection algorithm (Algorithm \ref{algo:main}), which is usable in any setting, for any three bodies of known coordinates and has a low computational cost. 
Running the algorithm for the Sun-Earth-Moon System on the trajectories obtained in the previous section, minute by minute, we quickly and correctly predict the two solar eclipses of 2026: the annular solar eclipse of 17th February 2026 (with a negligible error) and the total solar eclipse of the 12th of August 2026 (with an error ranging from approximately 11 to 13 minutes). This error was measured as the resulting difference between our observations and NASA's Goddard Space Flight Center's eclipse predictions \cite{nasa}.

This difference comes from the oversimplifications we considered for our Solar System and the approximation methods utilized by us (errors which are discussed in Section \ref{subsec:error}), as the purpose of this study was not to obtain a better eclipse prediction algorithm, but rather a didactical one, consisting in checking how much would the error be if we would simplify the problem of eclipse predictions to an undergraduate level.

Noteworthy is also the fact that in the current setting, the algorithm can detect forthcoming and historic eclipses for a bigger time-span, yet as the duration from the initial epoch considered increases, the error will also increase, hence the three-body problem is not suitable for long-term predictions; but for the purpose of short-term predictions in the range of one-two years it is well-equipped. Despite this time frame limitation, note also that given exact positions of any three bodies, the developed algorithms can precisely detect if an eclipse is occurring (with negligible error, only due to possible floating point arithmetic error).

%%%%%%%%%%%%%%%%%%%%%%%%%%%%%%%%%%%%%%%%%%%%%%%%%%%%%%%%%%%%%%%
%%%%%%%%%%%%%%%%%%%%%%%%%%%%%%%%%%%%%%%%%%%%%%%%%%%%%%%%%%%%%%%
%%%%%%%%%%%%%%%%%%%%%%%%%%%%%%%%%%%%%%%%%%%%%%%%%%%%%%%%%%%%%%%
\section{The geometry of eclipses}\label{sec:geometry}

This section is dedicated to creating the geometrical framework in which we can define an eclipse, detect it and classify it. As such, throughout this section, if not stated otherwise, our framework is the $\mathbb{R}^3$ Euclidean space endowed with the usual inner product.

%%%%%%%%%%%%%%%%%%%%%%%%%%%%%%%%%%%%%%%%%%%%%%%%%%%%%%%%%%%%%%%
\subsection{Cones and double cones}
\label{subsec:double_cone}

We begin by introducing the notion of a double cone, which will be useful later on to describe the so-called eclipse cones in which any observer can experience an eclipse.

\begin{definition} \label{def:convex_circular_cone}
    Consider a point $V\in\mathbb{R}^3$, a normalized vector $\vec{a}$ and an angle measure $\theta\in(0,\pi/2]$. The \textit{convex circular cone} $\mathcal{C}$ of vertex $V$, direction vector $\vec{a}$ and half-angle $\theta$ is the geometric locus
    \[\mathcal{C}=\left\{C\in\mathbb{R}^3:\angle(\overrightarrow{VC},\vec{a})\in[0,\theta]\right\}.\]
    We denote it by $\mathcal{C}=\mathrm{C}(V,\vec{a},\theta)$. If $a$ is a line passing through $V$ with direction vector $\vec{a}$, then $a$ is the \textit{symmetry axis} of $\mathcal{C}$. If we section $\mathcal{C}$ with a plane $\pi$ orthogonal to $a$, consider $\{O\}=\pi\cap a$. Then plane $\pi$ divides the convex circular cone $\mathcal{C}$ into a \textit{circular cone} of vertex $V$, direction vector $\vec{a}$, half-angle $\theta$ and base of center $O$; and a \textit{frustum} of $\mathcal{C}$ with base of center $O$.
\end{definition}

\begin{definition} \label{def:double_cone}
    Consider a convex circular cone $\mathcal{C}=\mathrm{C}(V,\vec{a},\theta)$ and $a$ its symmetry axis. Let $\mathcal{C}'=\mathrm{C}(V,-\vec{a},\theta)$ be the convex circular cone symmetric to $\mathcal{C}$ with respect to $V$. We call $\mathcal{C}\cup\mathcal{C}'$  the \textit{double cone} of vertex $V$, symmetry axis $a$ and half-angle $\theta$ and we denote it by 
    \[\mathrm{DC}(V,a,\theta).\]
\end{definition}

\begin{theorem} \label{th:double_cone}
    Consider a double cone $\mathcal{C}=\mathrm{DC}(V,a,\theta)$ and $\vec{a}$ a direction vector of line $a$. Then 
    \[\mathcal{C}=\left\{C\in\mathbb{R}^3:\left|\left\langle\vec{a},\overrightarrow{VC}\right\rangle\right|\geq \|\vec{a}\|\cdot \|\overrightarrow{VC}\|\cdot\cos\theta\right\}.\]
    In particular, the boundary of the double cone $\mathcal{C}$ is given by
    \[\mathrm{bd}\,\mathcal{C}=\left\{C\in\mathbb{R}^3:\left|\left\langle\vec{a},\overrightarrow{VC}\right\rangle\right|=\|\vec{a}\|\cdot \|\overrightarrow{VC}\|\cdot\cos\theta\right\}.\]
\end{theorem}
\begin{proof}
    By definition, a double cone $\mathcal{C}=\mathrm{DC}(V,a,\theta)$ is the geometric locus of all points $C$ for which the angle made by the direction vector $\vec{a}$ and by the vector $\overrightarrow{VC}$ belongs to the interval $[0,\theta]$ or $[\pi-\theta,\pi]$, i.e.
    \[\mathcal{C}=\left\{C\in\mathbb{R}^3:\angle(\vec{a},\overrightarrow{VC})\in[0,\theta]\cup[\pi-\theta,\pi]\right\}.\]

    For every point $C\in\mathbb{R}^3$, we know that
    \[\left|\left\langle\vec{a},\overrightarrow{VC}\right\rangle\right|=\|\vec{a}\|\cdot \|\overrightarrow{VC}\|\cdot \left|\cos\left(\angle(\vec{a},\overrightarrow{VC})\right)\right|.\]

    Therefore, $C\in\mathcal{C}$ if and only if 
    \[\left|\cos\left(\angle(\vec{a},\overrightarrow{VC})\right)\right|\geq \cos\theta,\]
    and as $\theta\in(0,\pi/2]$ and $\angle(\vec{a},\overrightarrow{VC})\in [0,\pi]$, it is clear that $C\in\mathcal{C}$ if and only if \[\angle(\vec{a},\overrightarrow{VC})\in[0,\theta]\cup[\pi-\theta,\pi].\qedhere\]
\end{proof}

The next question we want to answer about double cones is how to determine whether a double cone and a sphere intersect. Consider a spherical body $\P$ of center $P(x_P,y_P,z_P)$ and radius $R_P$ and a double cone $\mathcal{C}=\mathrm{DC}(V,a,\theta)$. Let $\vec{a}(x_a,y_a,z_a)$ be a direction vector of $a$. One approach to detect if $\P$ and $\mathcal{C}$ intersect is the following.

We define the distance between $\mathcal{C}$ and $\P$ in the following manner:
$$
\mathrm{dist}(\P,\mathcal{C}) = \mathrm{dist}(P,\mathrm{bd}\,\mathcal{C}) - R_P 
$$
where $\mathrm{dist}(P,\mathrm{bd}\,\mathcal{C})$ represents the distance between the point $P$ and the boundary of the double cone $\mathcal{C}$, i.e., 
$$
\mathrm{dist}(P,\mathcal{C}) =  \inf_{C \in \mathrm{bd}\mathcal{C}} \|P - C\|.
$$

In order to compute $\mathrm{dist}(P,\mathcal{C})$, we will minimize the following square distance function:
$$
f(x = (x_1,x_2,x_3)) = (x_1-x_P)^2+(x_2-y_P)^2+(x_3-z_P)^2$$ with respect to $\mathrm{bd}\,\mathcal{C}.$ Observe that $f$ is indefinitely differentiable and strictly convex. We denote by $v$ the vector $v = (x_V,y_V,z_V)$, and consider the function $F:\mathbb{R}^3 \to \mathbb{R}$, defined by
$$
F(x) = \langle \vec{a}, x-v \rangle ^2 - \|\vec{a}\|^2 \cdot \|x-v\|^2 \cdot \cos^2\theta.
$$

Observing that \[\mathrm{bd}\,\mathcal{C}=\left\{C\in\mathbb{R}^3: \langle\vec{a},\overrightarrow{VC}\rangle^2=\|\vec{a}\|^2\cdot \|\overrightarrow{VC}\| ^2\cdot\cos^2\theta\right\},\]  we deduce that, minimizing the function $f$ with respect to $\mathrm{bd}\,\mathcal{C}$, is equivalent to minimizing the function $f$ relative to the set $M = \{x\in \mathbb{R}^3 : F(x) = 0\}$. Therefore, we will use Lagrange multipliers method. The associated Lagrangian function is
$$
\mathcal{L}(x,\lambda) = f(x) + \lambda F(x),
$$
with the associated Lagrangian system
\begin{equation}\label{eq:Lagrange_system}\resizebox{0.9\textwidth}{!}{$
\left\{
\begin{aligned}
\frac{\partial\mathcal{L}}{\partial x_1}(x,\lambda)
&= -2(x_P - x_1)
  + \lambda\Bigl[2<\vec{a},x-v>\,x_a - 2\|\vec{a}\|^2\,(x_1-x_V)\cos^2\theta\Bigr]
=0\\
\frac{\partial\mathcal{L}}{\partial x_2}(x,\lambda)
&= -2(y_P - x_2)
  + \lambda\Bigl[2<\vec{a},x-v>\,y_a - 2\|\vec{a}\|^2\,(x_2-y_V)\cos^2\theta\Bigr]
=0\\
\frac{\partial\mathcal{L}}{\partial x_3}(x,\lambda)
&= -2(z_P - x_3)
  + \lambda\Bigl[2<\vec{a},x-v>\,z_a - 2\|\vec{a}\|^2\,(x_3-z_V)\cos^2\theta\Bigr]
=0\\
\frac{\partial\mathcal{L}}{\partial \lambda}(x,\lambda)
&= F(x) = 0.
\end{aligned}
\right.$}
\end{equation}

Solving the system \eqref{eq:Lagrange_system}, will give us a  solution (not necessarily unique, since the double cone $\mathcal{C}$ is not convex) $(x^0,\lambda^0)$, from which $x^0$ minimizes the square distance function $f$, relative to $\mathrm{bd}\,\mathcal{C}$, hence 
\[\mathrm{dist}(\P,\mathcal{C})=\sqrt{f(x^0)}-R_P.\]

Note that this distance can be zero or negative. In this case the intersection of $\P$ and $\mathcal{C}$ is nonempty. Otherwise, if the distance is strictly positive, then $\P$ and $\mathcal{C}$ are disjoint.

Implementing this method has proven to be computationally expensive, therefore in the following algorithm we will provide an alternative geometric approach, which has been considerably faster when using computer assisted computations.

\begin{algorithm}[Intersection between a sphere and a double cone] \label{algo:intersection_sphere_double_cone}
    Consider a spherical body $\P$ of center $P(x_P,y_P,z_P)$ and radius $R_P$ and a double cone $\mathcal{C}=\mathrm{DC}(V,a,\theta)$. Let $P'$ be the orthogonal projection of $P$ on line $a$. In order to determine, whether $\P$ and $\mathcal{C}$ are disjoint or not, it suffices to consider the following cases.
    \begin{itemize}
        \item If $P=P'$, then $P\in\mathrm{Int}\,\mathcal{C}$ and hence $\P$ and $\mathcal{C}$ are not disjoint;
        \item Else, the line $PP'$ will intersect $\mathcal{C}$ in two points. Let $P''$ be the one of the two for which the segment $PP''$ is of minimum length. We continue by treating the following two subcases.
        \begin{itemize}
            \item If $PP'\leq P'P''$, then $P\in\mathcal{C}$ and hence $\P$ and $\mathcal{C}$ are not disjoint;
            \item Else, $P\in\mathrm{Ext}\,\mathcal{C}$. Let $Q$ be the orthogonal projection of $P$ on line $VP''$. If $PQ\leq R_P$, then $\P$ and $\mathcal{C}$ are not disjoint; otherwise, $\P$ and $\mathcal{C}$ are disjoint.
        \end{itemize}
    \end{itemize}
\end{algorithm}
\begin{proof}
    We analyze the algorithm in two dimensions. Consider a section of our sphere and double cone, by a plane $\pi$ containing the line $a$ and the point $P$ (see Fig. \ref{fig:2D_intersection}).

    \begin{figure}[htb]
    \centering
    \includegraphics[width=0.50\textwidth]{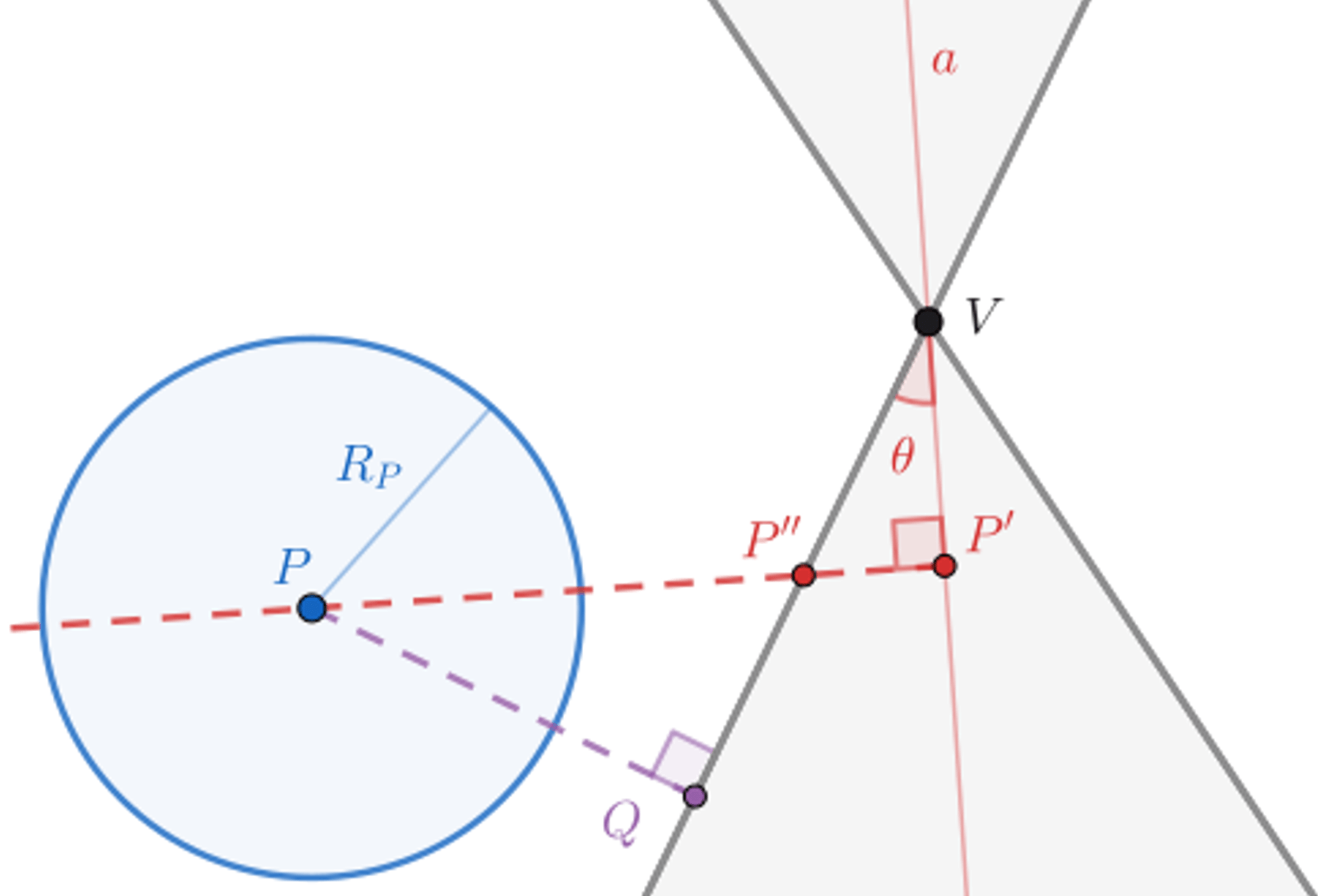}
    \caption{2D section of the double cone $\mathcal{C}$ and the sphere $\P$}
    \label{fig:2D_intersection}
    \end{figure}

    Let $\vec{a}(x_a,y_a,z_a)$ be a direction vector of line $a$.
    If $P'(x_{P'},y_{P'},z_{P'})$ is the orthogonal projection of $P(x_P,y_P,z_p)$ on the line $a$, then
    \[
    \left\{
    \begin{array}{rcl}
         x_{P'}&=&x_V-\alpha\cdot x_a \\[0.1cm]
         y_{P'}&=&y_V-\alpha\cdot y_a \\[0.1cm]
         z_{P'}&=&z_V-\alpha\cdot z_a, 
    \end{array}
    \right.
    \]
    where 
    \[\alpha=\frac{\left\langle\overrightarrow{VP},\vec{a}\right\rangle}{\|\vec{a}\|^2}.\]

    If $P=P'$, then $P\in a \subseteq \mathrm{Int}\,\mathcal{C}$ hence $P\in \P\cap\mathcal{C}$, therefore $\P$ and $\mathcal{C}$ are not disjoint.

    Assume that $P\neq P'$. The line $PP'$ will intersect $\mathcal{C}$ in two points. Let $P''$ be the one of the two for which the segment $PP''$ is of minimum length. 

    Given the right triangle $VP'P''$, we know that $\angle P''VP'=\theta$ hence we obtain the lengths
    \[P'P''=VP'\cdot \tan\theta\qquad\text{and}\qquad PP''=|PP'-P'P''|.\]
    
    If $PP'\leq P'P''$, then $P\in(P'P'']\subseteq\mathcal{C}$ hence $P\in\P\cap\mathcal{C}$, therefore $\P$ and $\mathcal{C}$ are not disjoint.

    Assume that $PP'> P'P''$, then $P''\in(PP')$, hence $P\in\mathrm{Ext}\,\mathcal{C}$. 
    
    Let $Q$ be the orthogonal projection of $P$ on line $VP''$, then the triangle $PQP''$ is a right triangle similar to the right triangle $VP'P''$, hence $\angle P''PQ=\theta$ and thus we obtain the length
    \[PQ=PP''\cdot\cos\theta.\]

    If $PQ\leq R_P$, then clearly $\P$ and $\mathcal{C}$ are not disjoint. Otherwise, if $PQ > R_P$, it is clear that $\P$ and $\mathcal{C}$ are disjoint.
\end{proof}

\begin{remark}
    The proof of Algorithm \ref{algo:intersection_sphere_double_cone}, not only proves the said algorithm, but also gives the formulas needed to compute the required elements of the algorithm. All the next algorithms of the current article will be given in the same manner.
\end{remark}

Finally, we also want to provide an answer to the question of when does a sphere intersect a convex circular cone, a circular cone and a frustum of a convex circular cone.

\begin{lemma} \label{lemma:same_semiplan}
    Let $A(x_A,y_A,z_A)$ and $H(x_H,y_H,z_H)$ be two points of the $\mathbb{R}^3$ Euclidean space. Consider $\mathcal{H}$ to be the plane orthogonal to the line $AH$, passing through the point $H$. The plane $\mathcal{H}$ determines two semispaces. An arbitrary point $B(x_B,y_B,z_B)\in\mathbb{R}^3$ is in the same semispace with point $A$ if and only if
    \[(x_A-x_H)(x_B-x_H)+(y_A-y_H)(y_B-y_H)+(z_A-z_H)(z_B-z_H) \geq 0.\]
\end{lemma}
\begin{proof}
    The proof is straightforward, given that for $A$ and $B$ to be in the same semispace relative to the plane $\mathcal{H}$, the angle $\angle AHB$ must be less than or equal to $90^\circ$. This condition is equivalent to 
    $$
    \langle \overrightarrow{HA}, \overrightarrow{HB}\rangle \geq 0.
    $$
    By applying the usual inner product rule, the conclusion follows.
\end{proof}

\begin{algorithm}[Intersection between a sphere and a convex circular cone] \label{algo:intersection_sphere_convex_cone}
Consider a spherical body $\P$ of center $P(x_P,y_P,z_P)$ and radius $R_P$, a convex circular cone $\mathcal{C}=\mathrm{C}(V,\vec{a},\theta)$, and $a$ to be the symmetry axis of $\mathcal{C}$. Let $\mathcal{DC}=\mathrm{DC}(V,a,\theta)$ be the double cone composed out of $\mathcal{C}$ and its symmetric with respect to its vertex $V$. We apply Algorithm \ref{algo:intersection_sphere_double_cone} to determine whether $\P$ intersects $\mathcal{DC}$:
\begin{itemize}
    \item If $\P$ and $\mathcal{DC}$ are disjoint, then it is clear that $\P$ and $\mathcal{C}$ are disjoint;
    \item If their intersection is nonempty, then we check the length of the segment $VP$:
    \begin{itemize}
        \item If $VP\leq R_P$, then it is clear that the sphere $\P$ intersects both convex cones forming $\mathcal{DC}$, therefore the intersection between $\P$ and $\mathcal{C}$ is nonempty;
        \item If $VP > R_P$, then the sphere $\P$ intersects only one of the two convex cones forming $\mathcal{DC}$. Consider $A$ to be point on line $a$ such that $\overrightarrow{VA}=\vec{a}$. By Lemma \ref{lemma:same_semiplan}, we check whether $A$ and $P$ are in the same semispace with respect to a plane passing through $V$, orthogonal to $a$, in which case the intersection between $\P$ and $\mathcal{C}$ is nonempty; otherwise they are disjoint.
    \end{itemize}
\end{itemize}  
\end{algorithm}
\begin{proof}
    The proof is straightforward.
\end{proof}

\begin{algorithm}[Intersection between a sphere and a frustum of a circular cone]\label{algo:intersection_sphere_frustum}
    Consider a spherical body $\P$ of center $P(x_P,y_P,z_P)$ and radius $R_P$, a convex circular cone $\mathcal{C}=\mathrm{C}(V,\vec{a},\theta)$, $a$ to be the symmetry axis of $\mathcal{C}$, and the point $O(x_O,y_O,z_O)\in a\cap \mathcal{C}$. 
    
    Let $\pi$ be the plane orthogonal to $a$ passing through $O$. Then plane $\pi$ will divide the convex circular cone $\mathcal{C}$ into a circular cone, denoted by $\mathcal{L}$; and a frustum of $\mathcal{C}$, denoted by $\mathcal{F}$. Note that both $\mathcal{L}$ and $\mathcal{F}$ have point $O$ as the center of their base. 
    
    We want to find if the sphere $\P$ intersects the frustum $\mathcal{F}$. We begin by applying Algorithm \ref{algo:intersection_sphere_convex_cone} (and implicitly Algorithm \ref{algo:intersection_sphere_double_cone}) to determine whether $\P$ intersects the convex circular cone $\mathcal{C}$. If $\P$ and $\mathcal{C}$ are disjoint, then so are $\P$ and $\mathcal{F}$.

    Assume that $\P$ and $\mathcal{C}$ are not disjoint. By Algorithm \ref{algo:intersection_sphere_double_cone} we have three cases:
    \begin{itemize}[leftmargin=*]
        \item If $P\in a$, the farthest point of $\P\cap \mathcal{C}$ from $V$ will be a point $F\in a$ such that $V,P$ and $F$ are collinear in this order, and that $PF=R_P$ (see Fig. \ref{fig:algo_frustum_sym}). We check if $VF\geq VO$, in which case $\P$ intersects $\mathcal{F}$. Otherwise $\P$ and $\mathcal{F}$ are disjoint.
        \item If $P\in\mathcal{C}$, the farthest point of $\P\cap \mathcal{C}$ from $V$ will be a point $F\in\mathrm{bd}\,\P\cap[VP$ such that $P\in [VF]$. Let $F'$ be the orthogonal projection of $F$ on line $a$  (see Fig. \ref{fig:algo_frustum_int}). We check if $VF'\geq VO$, in which case $\P$ intersects $\mathcal{F}$. Otherwise $\P$ and $\mathcal{F}$ are disjoint. 
        \item If $P\in\mathrm{Ext}\,\mathcal{C}$, then the farthest point of $\P\cap \mathcal{C}$ from $V$ will be a point $F\in \mathrm{bd}\,\P\cap \mathrm{bd}\,\mathcal{C}$ such that $VF$ is of maximum length. Let $F'$ be the orthogonal projection of $F$ on line $a$  (see Fig. \ref{fig:algo_frustum_ext}). We check if $VF'\geq VO$, in which case $\P$ intersects $\mathcal{F}$. Otherwise $\P$ and $\mathcal{F}$ are disjoint. 
    \end{itemize}
\end{algorithm}
\begin{figure}[htbp]
    \centering
    \hfill
    \begin{subfigure}[b]{0.3\textwidth}
        \centering
        \includegraphics[width=\textwidth, trim=10cm 5cm 14cm 1cm,
  clip]{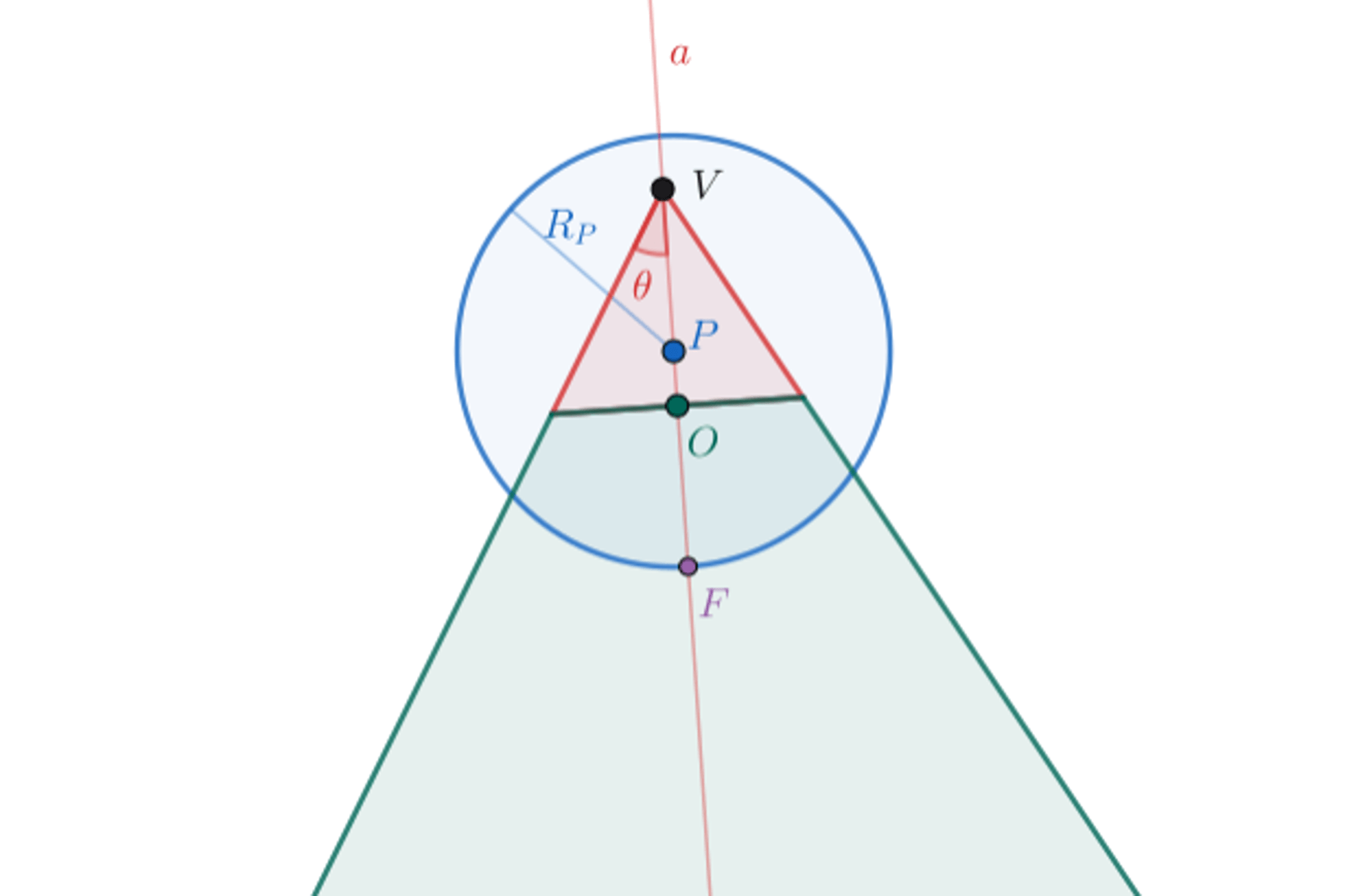}
        \caption{$P\in a$}
        \label{fig:algo_frustum_sym}
    \end{subfigure}
    \hfill
    \begin{subfigure}[b]{0.3\textwidth}
        \centering
        \includegraphics[width=\textwidth, trim=10cm 5cm 14cm 1cm,
  clip]{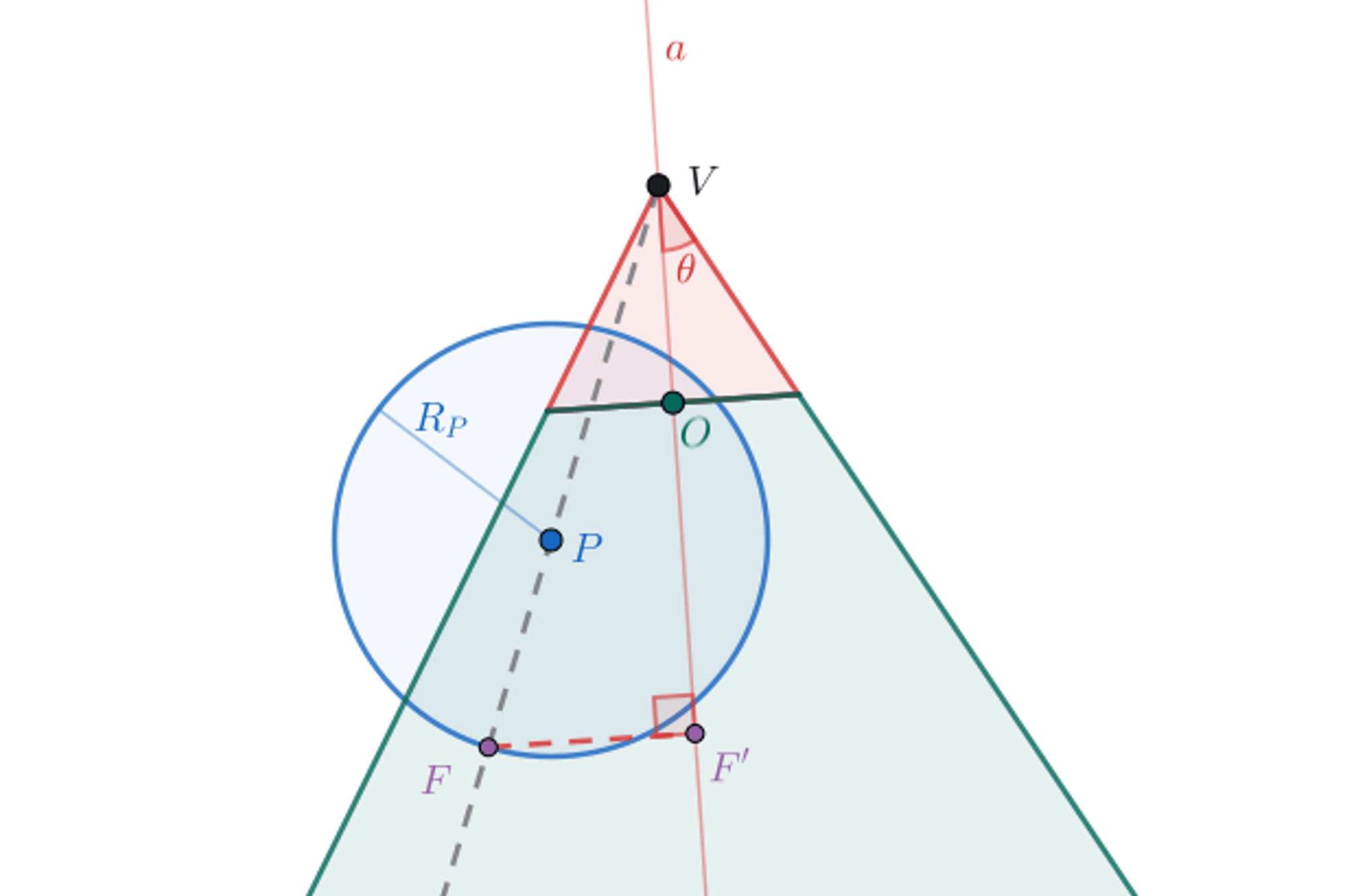}
        \caption{$P\in \mathcal{C}$}
        \label{fig:algo_frustum_int}
    \end{subfigure}
    \hfill
    \begin{subfigure}[b]{0.3\textwidth}
        \centering
        \includegraphics[width=\textwidth, trim=10cm 5cm 14cm 1cm,clip]{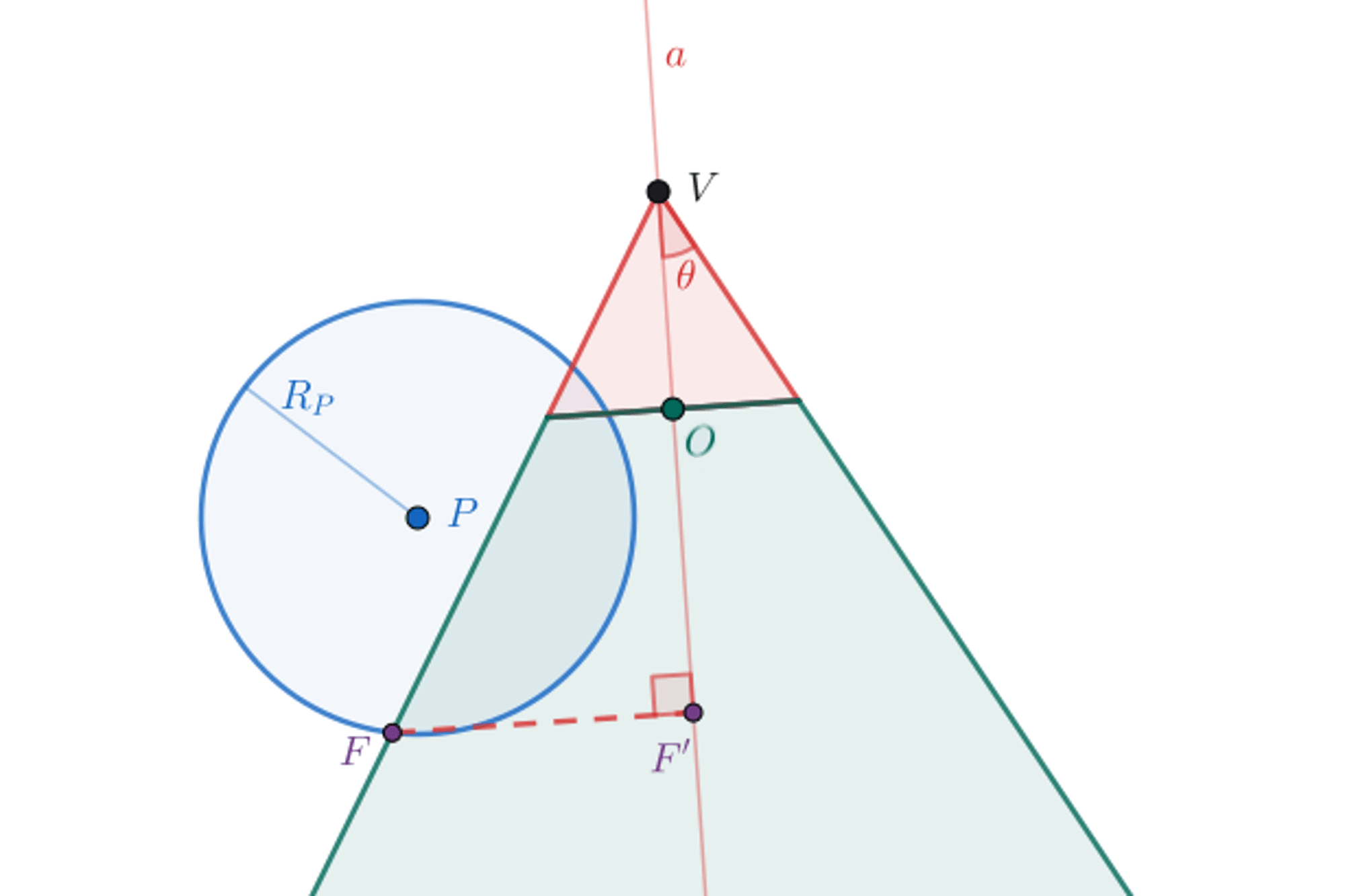}
        \caption{$P\in \mathrm{Ext}\,\mathcal{C}$}
        \label{fig:algo_frustum_ext}
    \end{subfigure}
    \hfill
    \label{fig:algorithm_frustum}
    \caption{2D section of the two bodies by a plane containing $a$ and $P$}
\end{figure}

\begin{proof}
    Assume that $\P$ and $\mathcal{F}$ are not disjoint. 
    If $P\in a$, and if $\vec{a}(x_{\vec{a}},y_{\vec{a}},z_{\vec{a}})$, then $$F(x_{V}+(PV+R_P)\cdot x_{\vec{a}},y_{V}+(PV+R_P)\cdot y_{\vec{a}},z_{V}+(PV+R_P)\cdot z_{\vec{a}}),$$
    because $\vec{a}$ is a normalized vector, collinear to $\overrightarrow{VF}$. If $VF\geq VO$, then $F\in\mathcal{F}$. Because $PF=R_P$, we also have that $F\in \P$, hence $F\in\P\cap\mathcal{F}$.

    If $P\in\mathcal{C}$, then $[VP\subseteq \mathcal{C}$ hence the farthest point of $\P\cap \mathcal{C}$ from $V$ is clearly the  point $F(x_F,y_F,z_F)\in \mathrm{bd}\,P\cap[VP$ such that $P\in [VF]$. Here,
    \[
    \left\{
    \begin{array}{rcl}
         x_F& = & x_V+\dfrac{VP+R_P}{VP}\cdot (x_P-x_V)  \\[0.3cm]
         y_F& = & y_V+\dfrac{VP+R_P}{VP}\cdot (y_P-y_V)  \\[0.3cm]
         z_F& = & z_V+\dfrac{VP+R_P}{VP}\cdot (z_P-z_V)  .
    \end{array}
    \right.
    \]
    Given that $F'(x_{F'},y_{F'},z_{F'})$ is the orthogonal projection of $F$ onto line $a$, line which is of normalized direction vector $\vec{a}$, we have that
    \[
    \left\{
    \begin{array}{rcl}
         x_{F'}& = & x_V-\left\langle\overrightarrow{VF},\vec{a}\right\rangle \cdot x_{\vec{a}} \\[0.2cm]
         y_{F'}& = & y_V- \left\langle\overrightarrow{VF},\vec{a}\right\rangle \cdot y_{\vec{a}} \\[0.2cm]
         z_{F'}& = & z_V- \left\langle\overrightarrow{VF},\vec{a}\right\rangle \cdot z_{\vec{a}} .
    \end{array}
    \right.
    \]
    If $VF'\geq VO$, then $F'\in \mathcal{F}$. Because $FF'\perp a$ and $\pi \perp a$, we get that $FF'\parallel \pi$ therefore, the fact that $F'\in \mathcal{F}$ and that $F\in[VP\subseteq \mathcal{C}$ implies $F\in \mathcal{F}$. This concludes that $F\in\P\cap\mathcal{F}$, hence $\P$ and $\mathcal{F}$ are not disjoint. Otherwise, if $VF'< VO$ we get that $F'\in \mathcal{L}$, hence, with the same reasoning, $F\in\mathcal{L}$. Because $F$ is the farthest point of $\P\cap \mathcal{C}$ from $V$, this clearly implies that $\P$ and $\mathcal{F}$ are disjoint.

    If $P\in\mathrm{Ext}\,\mathcal{C}$, then the farthest point of $\P\cap \mathcal{C}$ from $V$ will be a point $F\in \mathrm{bd}\,\P\cap \mathrm{bd}\,\mathcal{C}$ such that $VF$ is of maximum length. More precisely, if we consider a plane $\pi'$ that contains the point $P$ and the line $a$, then $\mathrm{bd}\,\P\cap \mathrm{bd}\,\mathcal{C}\cap \pi'$ contains two points (see Fig. \ref{fig:algo_frustum_ext}) - one closer to $V$ and one farther. Out of these two, $F$ will be the latter point. 

    \goodbreak
    The argument proceeds further as in the previous case, but with the mention that in this case we need not compute the coordinates of $F$ and $F'$. Instead, we shall make use of the distance $PQ$ computed by Algorithm \ref{algo:intersection_sphere_double_cone}, where $Q$ is the orthogonal projection of the point $P$ onto the nearest generatrix of the convex circular cone $\mathcal{C}$ to $\P$ (see Fig. \ref{fig:2D_intersection}). Applying the Pythagorean  theorem to the right triangles $VPQ$ and $PQF$ we get 
    \[VQ=\sqrt{PV^2-PQ^2}\quad\text{and}\quad QF=\sqrt{R_P^2-PQ^2},\]
    therefore
    \[VF=VQ+QF=\sqrt{PV^2-PQ^2}+\sqrt{R_P^2-PQ^2}.\]
    Finally, note that in the right triangle $VFF'$, $\angle FVF'=\theta$, hence \[VF'=VF\cdot\cos\theta.\qedhere\]
\end{proof}

\begin{algorithm}[Intersection between a sphere and a circular cone]\label{algo:intersection_sphere_cone}
    We work in the same setting as that of Algorithm \ref{algo:intersection_sphere_frustum}. We want to find if the sphere $\P$ intersects the circular cone $\mathcal{L}$. 
    
    We begin by applying Algorithm \ref{algo:intersection_sphere_convex_cone} (and implicitly Algorithm \ref{algo:intersection_sphere_double_cone}) to determine whether $\P$ intersects $\mathcal{C}$. If $\P$ and $\mathcal{C}$ are disjoint, then so are $\P$ and $\mathcal{L}$.
    
    Assume that $\P$ and $\mathcal{C}$ are not disjoint.
    We continue by computing the distance $PV$. If $PV \leq R_P$, then it is clear that $\P$ intersects the circular cone $\mathcal{L}$ (see Fig. \ref{fig:algo_cone_ini}). 
    Otherwise, by Algorithm \ref{algo:intersection_sphere_double_cone} we have three cases:
    \begin{itemize}
        \item If $P\in a$, the nearest point of $\P\cap \mathcal{C}$ to $V$ (i.e., the metric projection of $V$ onto the compact, convex set $\P\cap \mathcal{C}$) will be a point $N\in a$ such that $V,N$ and $P$ are collinear in this order, and that $PN=R_P$ (see Fig. \ref{fig:algo_cone_sym}). We check if $VN\leq VO$, in which case $\P$ intersects $\mathcal{L}$. Otherwise $\P$ and $\mathcal{L}$ are disjoint.
        
        \item If $P\in\mathcal{C}$, the nearest point of $\P\cap \mathcal{C}$ to $V$ will be a point $N\in\mathrm{bd}\,\P\cap[VP$ such that $N\in [VP]$. Let $N'$ be the orthogonal projection of $N$ on line $a$  (see Fig. \ref{fig:algo_cone_int}). We check if $VN'\leq VO$, in which case $\P$ intersects $\mathcal{L}$. Otherwise $\P$ and $\mathcal{L}$ are disjoint. 
        
        \item If $P\in\mathrm{Ext}\,\mathcal{C}$, then the nearest point of $\P\cap \mathcal{C}$ to $V$ will be a point $N\in \mathrm{bd}\,\P\cap \mathrm{bd}\,\mathcal{C}$ such that $VN$ is of minimum length. Let $N'$ be the orthogonal projection of $N$ on line $a$  (see Fig. \ref{fig:algo_cone_ext}). We check if $VN'\leq VO$, in which case $\P$ intersects $\mathcal{L}$. Otherwise $\P$ and $\mathcal{L}$ are disjoint. 
    \end{itemize}
\end{algorithm}

\begin{proof}
    Assume that the intersection of $\P$ and $\mathcal{C}$ is nonempty. 
    The comparison between the distance $PV$ and $R_P$ is necessary, in order to provide a similar algorithm to Algorithm \ref{algo:intersection_sphere_frustum}.
    Namely, if $PV\leq R_P$, then it is clear that $V\in \P$, hence $\P$ intersects $\mathcal{L}$ because $V\in\mathcal{L}$; otherwise, if we assume that $PV > R_P$, then we can compute the nearest point between $\P\cap \mathcal{C}$ to $V$ in a similar way to Algorithm \ref{algo:intersection_sphere_frustum}. Therefore from this point forward the proof is analogous to that of said algorithm, with the only differences lying in the computations of the main geometrical elements required. 
\end{proof}

\begin{figure}[hbtp]
    \centering
    \begin{subfigure}[b]{0.4\textwidth}
        \centering
        \includegraphics[width=\textwidth]{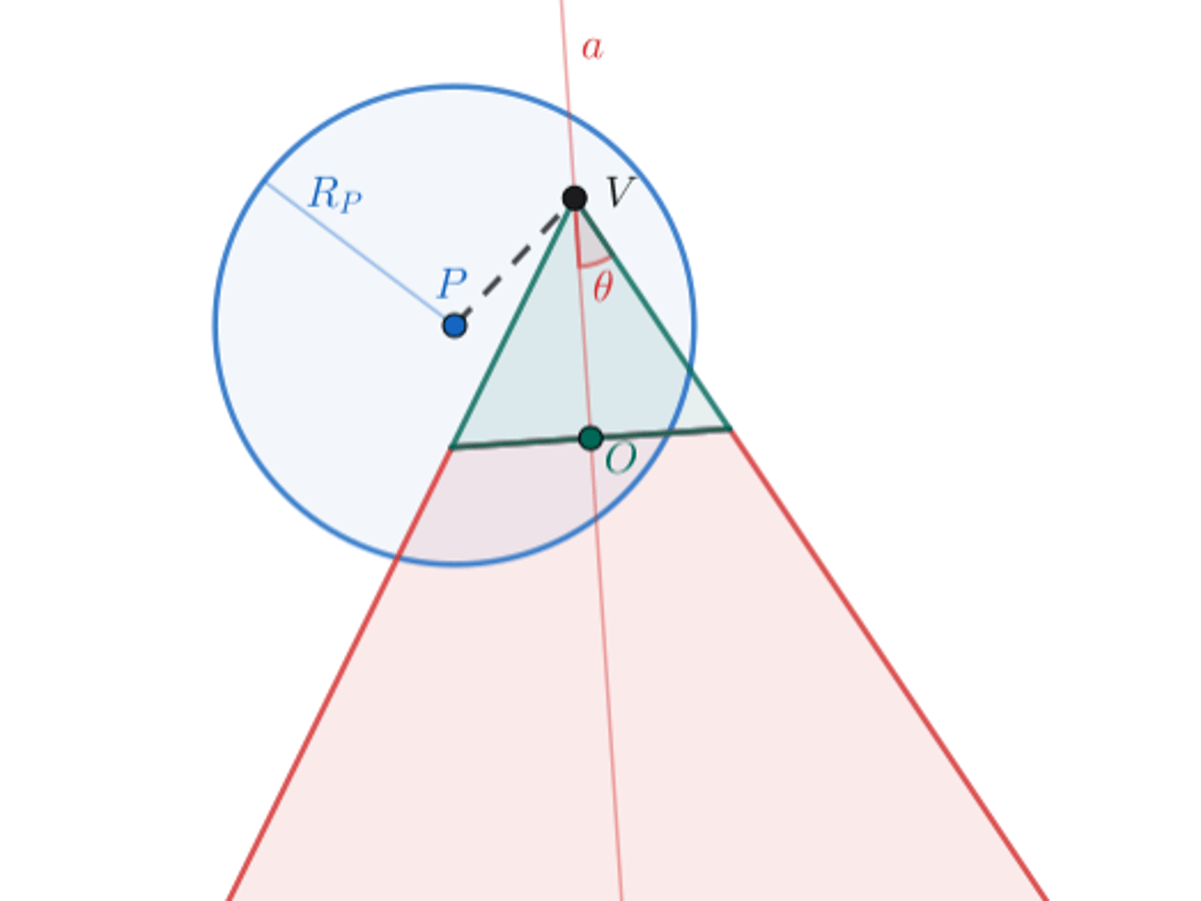}
        \caption{$PV\leq R_P$}
        \label{fig:algo_cone_ini}
    \end{subfigure}
    \begin{subfigure}[b]{0.4\textwidth}
        \centering
        \includegraphics[width=\textwidth]{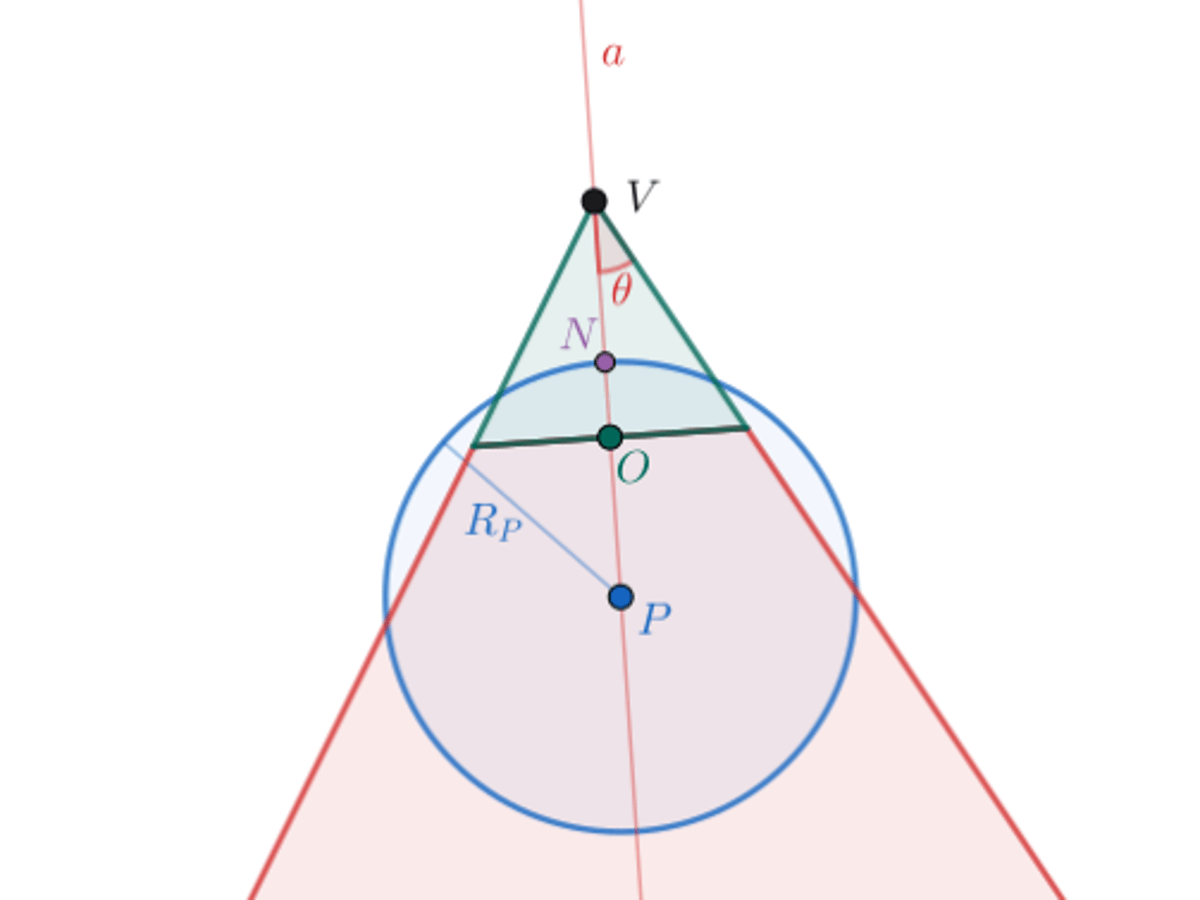}
        \caption{$P\in a$}
        \label{fig:algo_cone_sym}
    \end{subfigure}\\
    \begin{subfigure}[b]{0.4\textwidth}
        \centering
        \includegraphics[width=\textwidth]{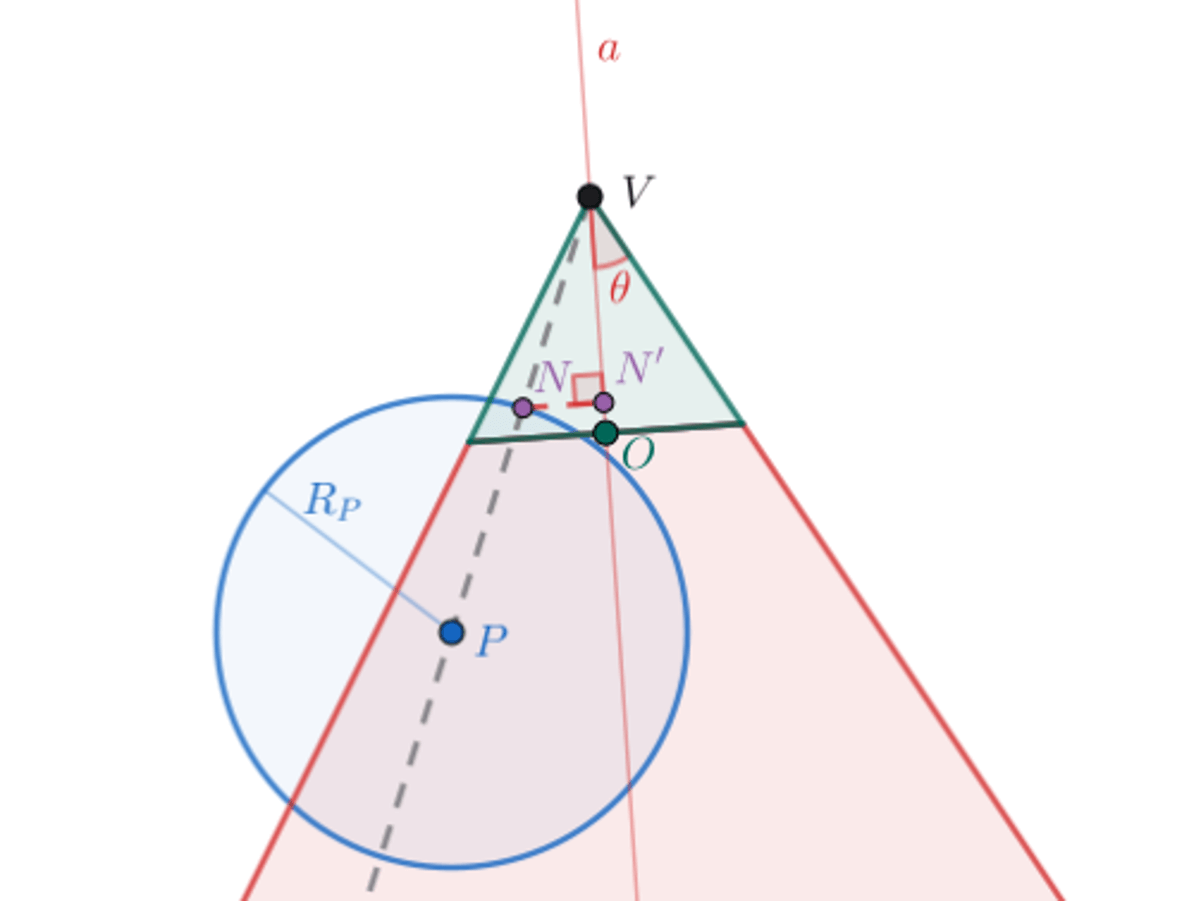}
        \caption{$P\in \mathcal{C}$}
        \label{fig:algo_cone_int}
    \end{subfigure}
    \begin{subfigure}[b]{0.4\textwidth}
        \centering
        \includegraphics[width=\textwidth]{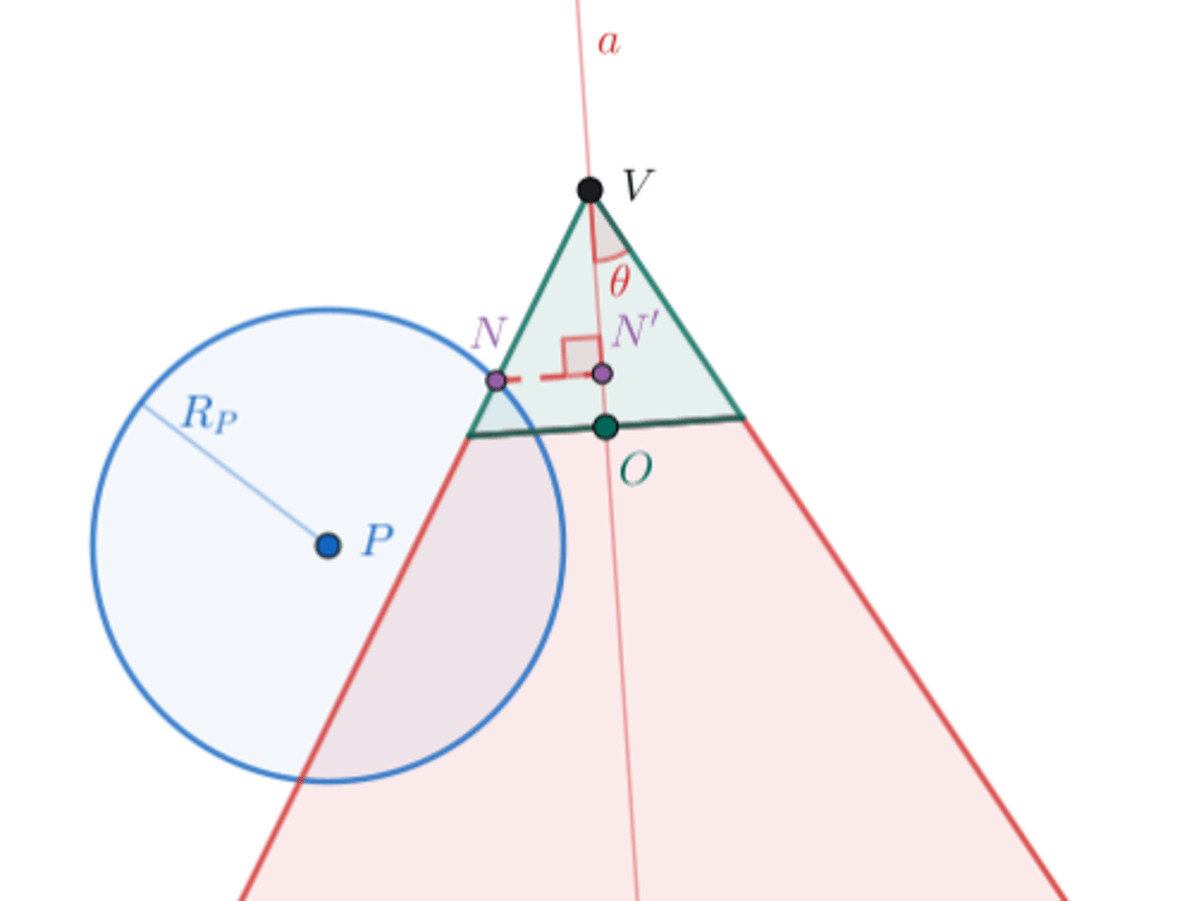}
        \caption{$P\in \mathrm{Ext}\,\mathcal{C}$}
        \label{fig:algo_cone_ext}
    \end{subfigure}
    \label{fig:algorithm_cone}
    \caption{2D section of the two bodies by a plane containing $a$ and $P$}
\end{figure}

\vspace{-0.2cm}
%%%%%%%%%%%%%%%%%%%%%%%%%%%%%%%%%%%%%%%%%%%%%%%%%%%%%%%%%%%%%%%
\subsection{Eclipse geometry}
\label{subsec:the_geometry_of_eclipses}

Consider a luminous spherical body $\Sun$ of radius $R_S$ and center $S(x_S,y_S,z_S)$; and an opaque spherical body $\M$ of radius $R_M$ and center $M(x_M,y_M,z_M)$. 

Assume that $R_S>R_M$, i.e. the luminous body is bigger than the opaque one.

Between the two spherical bodies $\Sun$ and $\M$ we have two types of tangents: external tangents and internal ones. These tangent lines will generate three noteworthy tangent cones:
\begin{itemize}[leftmargin=*]
\item The set of all exterior tangents between $\Sun$ and $\M$ form a double cone composed out of a convex circular cone that contains $\Sun$ and $\M$ in its interior and a convex circular cone with $\Sun$ and $\M$ in its exterior. The latter is called the \textit{antumbral cone} and we shall denote it by $\A$ (see Fig. \ref{fig:eclipse_umbral_and_antumbral_cones}); while the former can be divided by a plane containing all tangency points of the exterior tangents with $\M$, into a circular cone and a frustum. The circular cone is called the \textit{umbral cone} and we will denote it by $\U$ (see Fig. \ref{fig:eclipse_umbral_and_antumbral_cones}).

\item The set of all interior tangents between $\Sun$ and $\M$ form a double cone containing a convex circular cone that contains $\Sun$ in its interior and a convex circular cone that contains $\M$ in its interior. If we divide the latter by a plane containing all tangency points of the interior tangents with $\M$, into a circular cone and a frustum, the frustum, while not a circular cone, but rather a part of one, is called the \textit{penumbral cone} and we denote it by $\P$ (see Fig. \ref{fig:eclipse_penumbral_cone}).
\end{itemize}

Note that, although these three cones are not of the same geometrical type, we shall simply refer to them as \textit{eclipse cones}.

The names reflect the fact that for an observer inside either the umbral or antumbral cones, an \textit{eclipse} of $\Sun$ by $\M$ will be experienced; and for an observer inside the penumbral cone, a \textit{partial eclipse} of $\Sun$ by $\M$ will be experienced.

\begin{figure}[htbp]
    \centering
    \begin{subfigure}[b]{0.48\textwidth}
        \centering
        \includegraphics[width=\textwidth]{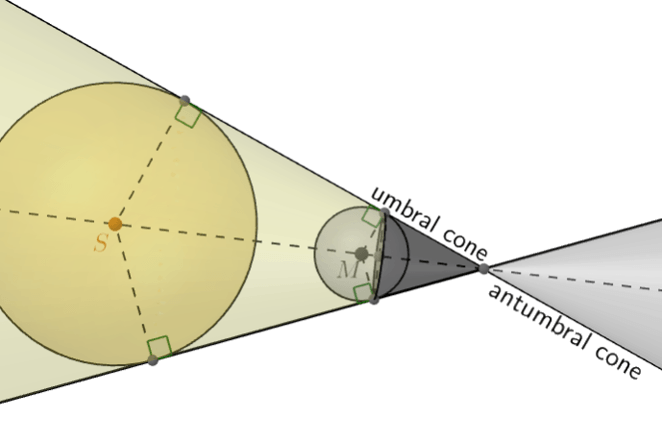}
        \caption{Umbral and antumbral cones}
        \label{fig:eclipse_umbral_and_antumbral_cones}
    \end{subfigure}
    \begin{subfigure}[b]{0.48\textwidth}
        \centering
        \includegraphics[width=\textwidth]{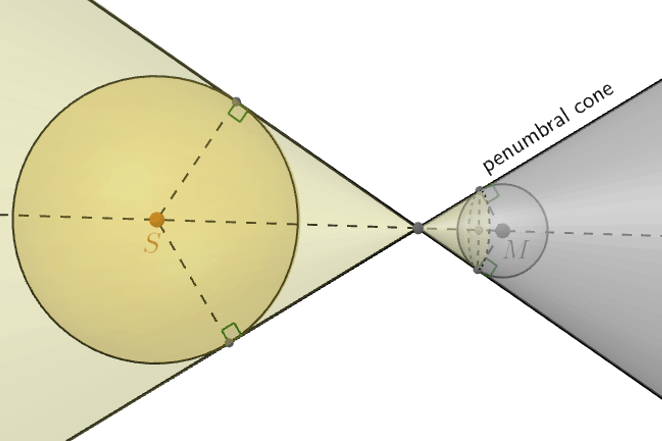}
        \caption{Penumbral cone}
        \label{fig:eclipse_penumbral_cone}
    \end{subfigure}
    \label{fig:eclipse_3Dintro}
    \caption{Eclipse cones}
\end{figure}

We will denote the common vertex of the umbral and antumbral cone by $V_\U=V_\A$ and the vertex of the penumbral cone by $V_\P$. Note that all three cones share the same symmetry axis, i.e. the line containing the points $S$, $M$, $V_\U$, and $V_\P$; and we will denote its normalized direction vector $\overrightarrow{SM}/SM$ with $\vec{a}$. The half-angle of the penumbral cone will be denoted by $\theta_\P$ and similarly we will denote by $\theta_\U$ the shared half-angle of both the umbral and antumbral cones. Finally, we will denote with $O_\U$ the center of the base of the umbral cone $\U$, and with $O_\P$ the center of the base of the penumbral cone $\P$.

In what follows we proceed to obtain formulas for the coordinates of $V_\U$, $V_\P$, $\vec{a}$, $O_\U$, $O_\P$, and for the angles $\theta_\U$, $\theta_\P$, in terms of the parameters of the spheres $\Sun$ and $\M$.

\begin{lemma}\label{lemma:direction_vector_of_cones}
    If $\vec{a}(x_{\vec{a}},y_{\vec{a}},z_{\vec{a}})$ is the direction vector of the antumbral and penumbral cones, then
    \[\vec{a}\left(\frac{x_M-x_S}{SM},\frac{y_M-y_S}{SM},\frac{z_M-z_S}{SM}\right).\]
    
    The direction vector of the umbral cone is $-\vec{a}$.
\end{lemma}
\begin{proof}
    It is straightforward, given that $\vec{a}=\overrightarrow{SM}/SM$.
\end{proof}

\begin{lemma}\label{Lemma:coordonatevertex}
    If $V_\U(x_{V_\U},y_{V_\U},z_{V_\U})$ is the common vertex of the umbral and antumbral cones then
   \begin{equation}
        \left\{
            \begin{array}{l}
                x_{V_\U} = \dfrac{R_S}{R_S - R_M} x_M - \dfrac{R_M}{R_S - R_M} x_S \\[10pt]
                y_{V_\U} =  \dfrac{R_S}{R_S - R_M}  y_M - \dfrac{R_M}{R_S - R_M} y_S \\[10pt]
                z_{V_\U} = \dfrac{R_S}{R_S - R_M}  z_M - \dfrac{R_M}{R_S - R_M} z_S.
            \end{array}
        \right.
        \label{eq:coordVfinal}
    \end{equation}
\end{lemma}
\begin{proof}
We consider a section of our spheres and umbral cone, by a plane $\pi$ containing the line $SM$ (see Fig. \ref{fig:eclipse_2Dintro}). Furthermore, we denote by $T_{S1}$, $T_{S2}$, $T_{M1}$ and $T_{M2}$, the points of tangency of the external common tangent lines to the circles $C(S,R_S)$ and $C(M,R_M)$, respectively.

\begin{figure}[htbp]
    \centering
    \includegraphics[width=0.55\textwidth]{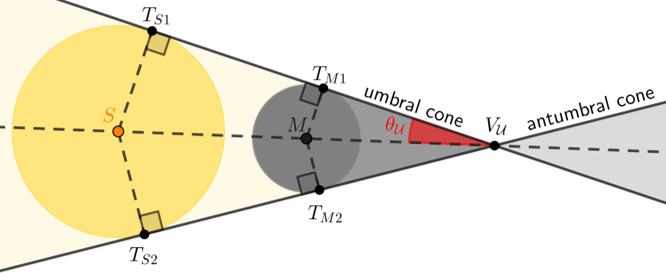}
    \caption{2D section of the two bodies}
    \label{fig:eclipse_2Dintro}
\end{figure}

Since
$
ST_{S1} \parallel MT_{M1},
$
by the fundamental theorem of similarity, we have that $\triangle V_\U MT_{M1} \sim \triangle V_\U ST_{S1}.$ Thus, we obtain 
$$
\frac{V_\U M}{V_\U S}
=
\frac{MT_{M1}}{ST_{S1}}
=
\frac{V_\U T_{M1}}{V_\U T_{S1}}.
$$

Note that \(MT_{M1}\) is equal to radius of $\M$, \(R_M\), and
\(ST_{S1}\) with the radius of $\Sun$, \(R_S\), therefore, we obtain:
$$
\frac{V_\U M}{V_\U S}
=
\frac{R_M}{R_S}
$$
therefore,
\begin{equation*}
V_\U M
=
\frac{R_M}{R_S - R_M}
\cdot MS.
\end{equation*}

By observing that the points $S$, $M$ and $V_\U$ are collinear, we deduce the existence of a scalar $\lambda \in \mathbb{R}$  such that 
\begin{equation}\label{eq:V_U_deduce}
\overrightarrow{MV_\U} = \lambda \overrightarrow{SM}.\end{equation} 

Since
\[
\left|\overrightarrow{MV_\U}\right|
=
\frac{R_M}{R_S-R_M}
\cdot
\left|\overrightarrow{MS}\right|
\]
we get
\[
|\lambda|
=
\frac{
\dfrac{R_M}{R_S-R_M}
\cdot
\left|\overrightarrow{MS}\right|
}
{\left|\overrightarrow{SM}\right|} = \frac{R_M}{R_S-R_M}.
\]

Moreover, the fact that the vectors \(\overrightarrow{MV_\U}\) and \(\overrightarrow{SM}\) are pointing in the same direction
implies that \(\lambda>0\). Therefore, $\lambda = \frac{R_M}{R_S-R_M}.$ By \eqref{eq:V_U_deduce}, we get
       \begin{equation*}
        \left\{
            \begin{array}{l}
                x_{V_\U} = \left(1+\dfrac{R_M}{R_S - R_M}\right) x_M - \dfrac{R_M}{R_S - R_M} x_S \\[10pt]
                y_{V_\U} =  \left(1+\dfrac{R_M}{R_S - R_M}\right)  y_M - \dfrac{R_M}{R_S - R_M} y_S \\[10pt]
                z_{V_\U} = \left(1+\dfrac{R_M}{R_S - R_M}\right)z_M - \dfrac{R_M}{R_S - R_M} z_S,
            \end{array}
        \right.
    \end{equation*}
    and the conclusion follows.
\end{proof}

\begin{lemma}\label{lemma:vertex_of_penumbral_cone}
    If $V_\P(x_{V_\P},y_{V_\P},z_{V_\P})$ is the vertex of the penumbral cone then
   \begin{equation}
        \left\{
            \begin{array}{l}
                x_{V_\P} = \dfrac{R_S}{R_S + R_M} x_M + \dfrac{R_M}{R_S + R_M} x_S \\[10pt]
                y_{V_\P} =  \dfrac{R_S}{R_S + R_M}  y_M + \dfrac{R_M}{R_S + R_M} y_S \\[10pt]
                z_{V_\P} = \dfrac{R_S}{R_S + R_M}  z_M + \dfrac{R_M}{R_S + R_M} z_S.
            \end{array}
        \right.
        \label{eq:coordVP}
    \end{equation}
\end{lemma}
\begin{proof}
If follows by an analogous argument as in Lemma \ref{Lemma:coordonatevertex}.
\end{proof}

\begin{lemma}\label{lemma:half-angles_of_cones}
Let $\theta_\U$ be the shared half-angle of both the umbral and antumbral cones, and let $\theta_\P$ be the half-angle of penumbral cone. Then 
$$
\theta_\U=\arcsin\left(\frac{R_S-R_M}{SM}\right)
\qquad\text{and}\qquad
\theta_\P=\arcsin\left(\frac{R_S+R_M}{SM}\right).
$$
\end{lemma}

\begin{proof}
We continue with the same setting as in the proof of Lemma \ref{Lemma:coordonatevertex}. If $K$ is the foot of perpendicular from $M$ to $ST_{S_1}$ (see Fig. \ref{fig:eclipse_2D_theta_u}), then, in the right triangle $MKS$ we have,
$$\theta_\U=\arcsin\left(\frac{R_S-R_M}{SM}\right).$$

\begin{figure}[htbp]
    \centering
    \begin{subfigure}[b]{0.48\textwidth}
        \centering
        \includegraphics[width=\textwidth]{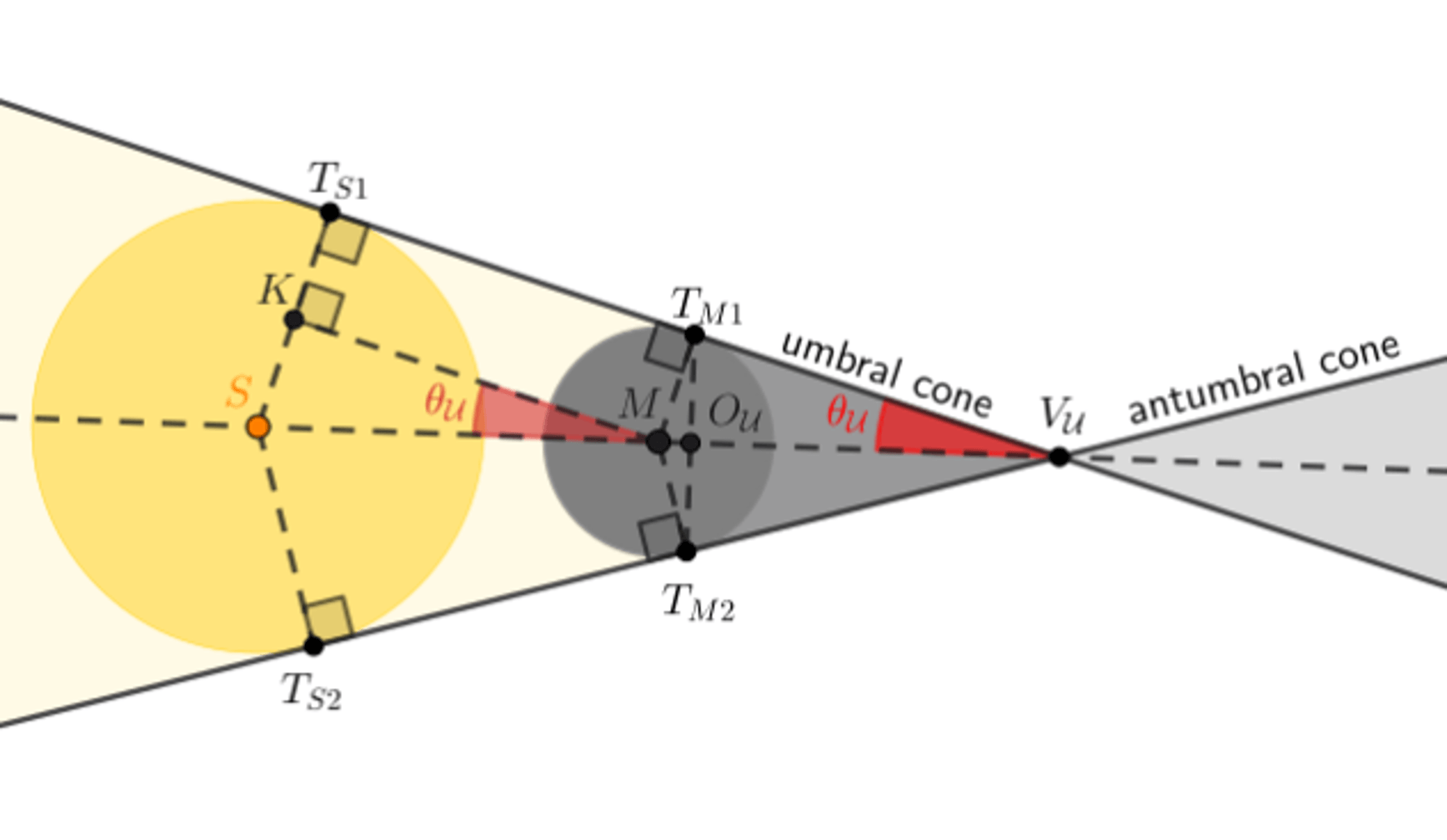}
        \caption{The angle $\theta_\U$}
        \label{fig:eclipse_2D_theta_u}
    \end{subfigure}
    \hfill
    \begin{subfigure}[b]{0.48\textwidth}
        \centering
        \includegraphics[width=\textwidth]{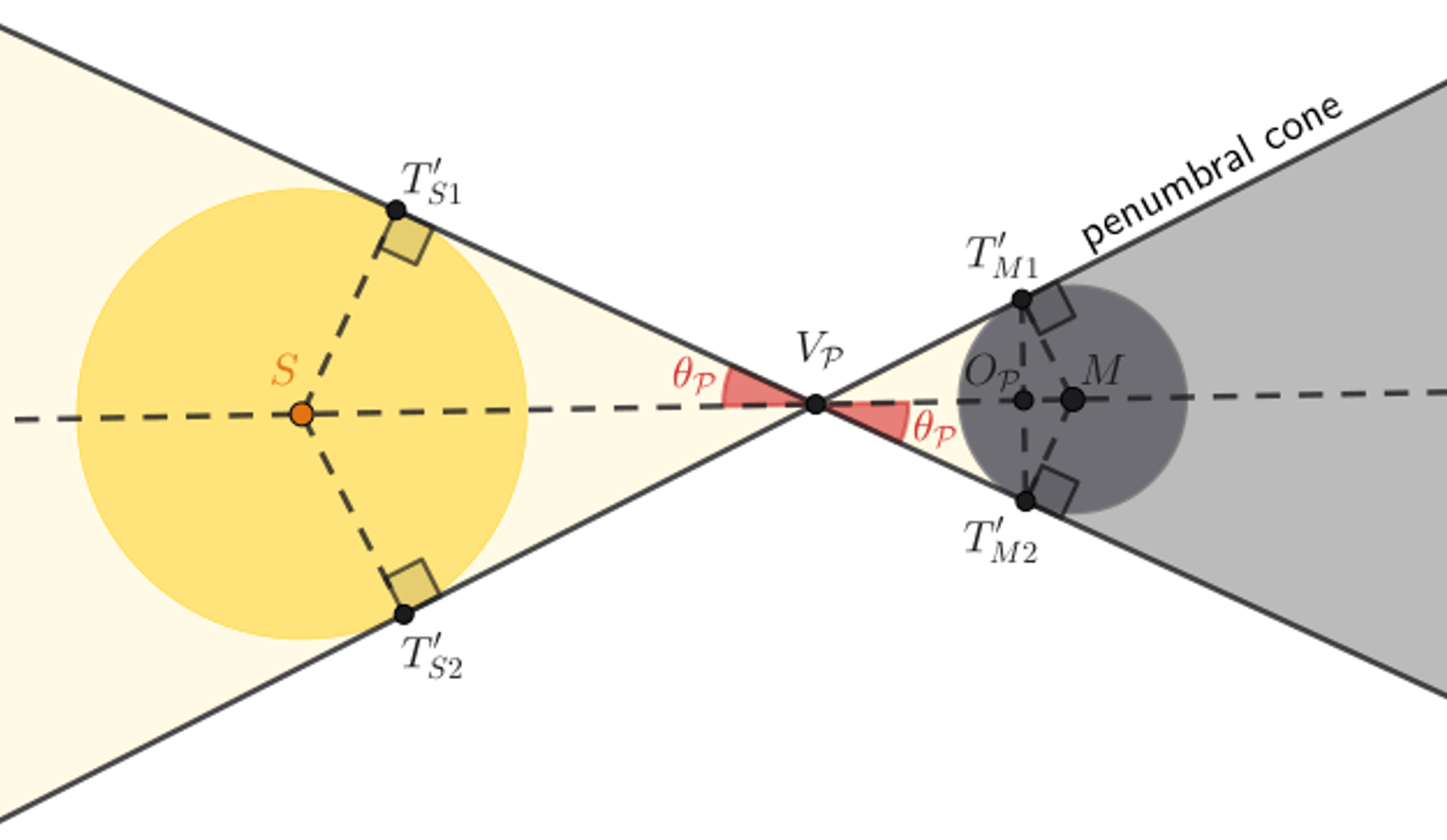}
        \caption{The angle $\theta_\P$}
        \label{fig:eclipse_2D_theta_p}
    \end{subfigure}
    \caption{2D section of the two bodies}
    \label{fig:eclipse_2D_angles}
\end{figure}

In order to prove the second formula, we denote by $T'_{S1}, T'_{S2}, T'_{M1}$ and $ T'_{M2}$, the points of tangency of the penumbral tangent lines with $\Sun$ and $\M$, respectively (see Fig. \ref{fig:eclipse_2D_theta_p}).

Applying the sine function in the right triangles $V_\P T'_{S1}S$ and $V_\P T'_{M2}M$ we get  
$$\sin(\theta_\P)=\frac{R_S}{V_\P S} = \frac{R_M}{V_\P M} =\frac{R_S+R_M}{V_\P S + V_\P M} = \frac{R_S+R_M}{SM}.$$ The conclusion follows immediately. 
\end{proof}

\begin{lemma} \label{lemma:heights_of_cones}
    If $O_\U(x_{O_\U},y_{O_\U},z_{O_\U})$ is the center of the base of the umbral cone $\U$, then
    \[
    |V_{\U}O_{\U}|=R_M\frac{\cos^2(\theta_\U)}{\sin(\theta_\U)}\quad\text{ and }\quad
    \left\{
    \begin{array}{rcl}
         x_{O_{\U}}&=&x_{V_{\U}}-|V_{\U}O_{\U}|\cdot x_{\vec{a}}\\
         y_{O_{\U}}&=&y_{V_{\U}}-|V_{\U}O_{\U}|\cdot y_{\vec{a}}\\
         z_{O_{\U}}&=&z_{V_{\U}}-|V_{\U}O_{\U}|\cdot z_{\vec{a}}.
    \end{array}
    \right.
    \]
    Similarly, if $O_{\P}(x_{O_\P},y_{O_\P},z_{O_\P})$ is the center of the base of the penumbral cone $\P$, then
    \[
    |V_{\P}O_{\P}|=R_M\frac{\cos^2(\theta_\P)}{\sin(\theta_\P)}\quad\text{ and }\quad
    \left\{
    \begin{array}{rcl}
         x_{O_{\P}}&=&x_{V_{\P}}+|V_{\P}O_{\P}|\cdot x_{\vec{a}}\\
         y_{O_{\P}}&=&y_{V_{\P}}+|V_{\P}O_{\P}|\cdot y_{\vec{a}}\\
         z_{O_{\P}}&=&z_{V_{\P}}+|V_{\P}O_{\P}|\cdot z_{\vec{a}}.
    \end{array}
    \right.
    \]
\end{lemma}
\begin{proof}
    We continue in the same setting as in the proof of Lemma \ref{lemma:half-angles_of_cones} (see Fig. \ref{fig:eclipse_2D_angles}). 

    In the right triangle $V_\U T_{M1}M$, by the cathetus theorem we have $V_\U O_\U\cdot V_\U M=V_\U T_{M1}^2$, and by the Pythagorean theorem $V_\U T_{M1}^2=V_\U M^2-R_M^2$, hence
    \begin{equation}\label{eq:VO_lengths}
        V_\U O_\U=\frac{V_\U M^2-R_M^2}{V_\U M}.
    \end{equation}
    Moreover, in the same triangle, we have that the angle $\angle MV_\U T_{M_1}=\theta_\U$, therefore $\sin(\theta_\U)=R_M/V_\U M$, which implies $V_\U M=R_M/\sin(\theta_\U)$. Substituting in \eqref{eq:VO_lengths}, we get
     \[
    V_{\U}O_{\U}=R_M\frac{\cos^2(\theta_\U)}{\sin(\theta_\U)}.
    \]
    Finally, given that $|\vec{a}|=1$ and that $\overrightarrow{V_\U O_\U}$ and $\vec{a}$ are collinear, we obtain $\overrightarrow{V_\U O_\U}=-|V_\U O_\U|\cdot \vec{a}$ and the conclusion follows immediately.
\end{proof}

\begin{remark}\label{remark:eclipse_parameters}
    Given the parameters of the two spheres $\Sun$ and $\M$, by Lemma \ref{lemma:direction_vector_of_cones}, Lemma \ref{Lemma:coordonatevertex}, Lemma \ref{lemma:vertex_of_penumbral_cone}, Lemma \ref{lemma:half-angles_of_cones} and Lemma \ref{lemma:heights_of_cones}, we now know the parameters of the eclipse cones, i.e.:
    \begin{itemize}
        \item the umbral cone is the cone obtained via an orthogonal section of $\mathrm{C}(V_\U,-\vec{a},\theta_\U)$ at $O_\U$;
        \item the antumbral cone is $\A=\mathrm{C}(V_\U,\vec{a},\theta_\U)$;
        \item the penumbral cone is the frustum of $\P=\mathrm{C}(V_\P,\vec{a},\theta_\P)$, obtained via an orthogonal section at $O_\P$.
    \end{itemize}
    
    Moreover, if $a$ is a line passing through $V_\U$ and with direction vector $\vec{a}$, then to the double cone formed by $\A$ and $\mathrm{C}(V_\U,-\vec{a},\theta_\U)$ we will refer to as the \textit{umbral double cone} and we will denote it by
    \[\mathcal{DU}=\mathrm{DC}(V_\U,a,\theta_U).\]
    Similarly, the double cone obtained from $\P$ and its symmetric with respect to its vertex, will be called the \textit{penumbral double cone} and we will denote it by
    \[\mathcal{DP}=\mathrm{DC}(V_\P,a,\theta_P).\]
\end{remark}

%%%%%%%%%%%%%%%%%%%%%%%%%%%%%%%%%%%%%%%%%%%%%%%%%%%%%%%%%%%%%%%
\subsection{Annular vs. total eclipses} \label{subsec:annular_vs_total_eclipses}

In order to decide if an eclipse is annular or total, we must first recall the concept of angular diameter and its well-known formula.

\begin{definition} \label{def:angular_diameter}
The \textit{angular diameter} (also known as the apparent size)  of a spherical celestial body relative to an observer is the measure of the angle under which the observer sees the respective celestial body. 
\end{definition}

\begin{proposition}\label{prop:computation_of_angular_diameter}
    Consider $\Sun$ to be a spherical body of center $S$ and radius $R_S$; and $P$ to be an observer (represented as a point at a distance of $D=PS$ from $S$). Then the angular diameter of $\Sun$ relative to the observer $P$ is 
    \[
        \theta
        =
        2\cdot\arctan\left(
        \frac{R_S}{\sqrt{D^2-R_S^2}}
        \right).
    \]
\end{proposition}

\begin{remark} \label{remark:total_vs_annular_eclipse}
The type of eclipses experienced by an observer depends on the relationship between the angular diameters of the two bodies producing the eclipse. If we consider again the luminous spherical body $\Sun$ and the opaque spherical body $\M$ of Section \ref{subsec:the_geometry_of_eclipses}, and if we denote their corresponding angular diameters by $\theta_\Sun$ and $\theta_\M$, respectively (computed via Proposition \ref{prop:computation_of_angular_diameter}), then the following affirmations hold:
\begin{itemize}
    \item If 
    $
    \theta_\M\geq \theta_\Sun,$
     then a \textit{total eclipse} occurs.
    \item If
    $
    \theta_\M<\theta_\Sun,
    $
    then $\M$ cannot completely ``cover'' $\Sun$, resulting in an \textit{annular eclipse}.
\end{itemize}
\end{remark}

%%%%%%%%%%%%%%%%%%%%%%%%%%%%%%%%%%%%%%%%%%%%%%%%%%%%%%%%%%%%%%%
%%%%%%%%%%%%%%%%%%%%%%%%%%%%%%%%%%%%%%%%%%%%%%%%%%%%%%%%%%%%%%%
%%%%%%%%%%%%%%%%%%%%%%%%%%%%%%%%%%%%%%%%%%%%%%%%%%%%%%%%%%%%%%%
\section{The three-body problem}
\label{sec:3bps}

In order to determine if an eclipse is occurring over an extended period of time we need to know the positions of our celestial bodies (as we saw in the previous section). To obtain them, we need to analyze how the bodies interact gravitationally with each other, i.e. the $n$-body problem ($n\in\mathbb{N}^\ast$). This problem describes the motion of $n$ bodies as solutions of a system of differential equations (equations of motion) given their initial position and velocity. To derive the equations of motion for an $n$-body system, one may follow the approach  presented in \cite[\S 2.3]{book:Valtonen2005}.

In our paper, we will restrict the problem to the case of the Sun-Earth-Moon system, hence we will only deal with the three-body problem. In this context, we consider three point masses interacting gravitationally in a fixed Cartesian reference frame $Oxyz$, of the Euclidean space $\mathbb{R}^3$; with position vectors $\mathbf{r}_1\colon \mathbb{R} \to \mathbb{R}^3$, $\mathbf{r}_2\colon \mathbb{R} \to \mathbb{R}^3$, and $\mathbf{r}_3\colon \mathbb{R} \to \mathbb{R}^3$, and masses $m_1$, $m_2$, and $m_3$, respectively. 

The second-order vector differential equations of motion for the three-body system is:
\begin{equation}\label{eq:motion_equations_3_body_problem}
\left\{
\begin{array}{rcl}
\displaystyle \frac{\d ^2 \mathbf{r}_1}{\d t^2} &=& \displaystyle -G \left( \frac{m_2 (\mathbf{r}_1 - \mathbf{r}_2)}{|\mathbf{r}_1 - \mathbf{r}_2|^3} + \frac{m_3 (\mathbf{r}_1 - \mathbf{r}_3)}{|\mathbf{r}_1 - \mathbf{r}_3|^3} \right) \\[0.5cm]
\displaystyle \frac{\d ^2 \mathbf{r}_2}{\d t^2} &=& \displaystyle -G \left( \frac{m_1 (\mathbf{r}_2 - \mathbf{r}_1)}{|\mathbf{r}_2 - \mathbf{r}_1|^3} + \frac{m_3 (\mathbf{r}_2 - \mathbf{r}_3)}{|\mathbf{r}_2 - \mathbf{r}_3|^3} \right) \\[0.5cm]
\displaystyle \frac{\d ^2 \mathbf{r}_3}{\d t^2} &=& \displaystyle -G \left( \frac{m_1 (\mathbf{r}_3 - \mathbf{r}_1)}{|\mathbf{r}_3 - \mathbf{r}_1|^3} + \frac{m_2 (\mathbf{r}_3 - \mathbf{r}_2)}{|\mathbf{r}_3 - \mathbf{r}_2|^3} \right),
\end{array}
\right.
\end{equation}
where $G$ is the universal gravitational constant. According to the Committee on Data for Science and Technology (CODATA) of The National Institute of Standards and Technology (NIST) recommended values of the fundamental physical constants of 2022 \cite{art:CODATA}, $G$ has the approximate value:
\begin{equation} \label{eq:valoareG}
    G \approx 6.67430 \times 10^{-11} \; \frac{\text{m}^3}{\text{kg} \cdot \text{s}^2}.
\end{equation}

To solve system \eqref{eq:motion_equations_3_body_problem}, one also requires some initial values, hence obtaining a so-called initial value problem (IVP). Even so, finding an explicit closed-form solution is extremely difficult: Henri Poincaré demonstrated the complex and highly sensitive nature of the problem, revealing fundamental difficulties in obtaining general explicit solutions. Even small perturbations in the initial values may lead to significantly different dynamical evolutions, particularly over long integration times. That being said, using modern numerical methods alongside computer-assisted computations, one can obtain accurate short-term approximate solutions. One of the most popular methods to use, is the Runge-Kutta-Fehlberg method (RKF45) for the approximation of the solutions of ordinary differential equations.

%%%%%%%%%%%%%%%%%%%%%%%%%%%%%%%%%%%%%%%%%%%%%%%%%%%%%%%%%%%%%%%
\subsection{Data gathering}\label{subsec:data_gathering}

The equations of motion form a system of second-order ordinary differential equations \eqref{eq:motion_equations_3_body_problem}, and consequently, a well-posed IVP requires the specification of both the position vectors and the velocity vectors of all bodies at a common initial epoch. 

The initial epoch is fixed to $t_0 = \text{1 January 2026, 00:00 UTC}$. All position and velocity vectors are expressed in a Cartesian coordinate system relative to the \emph{Solar System Barycenter}. This barycentric reference frame is standard in celestial mechanics and ensures consistency with modern ephemerids.

The initial conditions for the selected celestial bodies (our analysis is restricted to the Sun-Earth-Moon system) are obtained from the \emph{NASA Jet Propulsion Laboratory (JPL) HORIZONS On-Line Ephemeris System} \cite{JPL:HORIZONS}. This service provides high-precision state vectors computed using observational data and validated dynamical models.

To generate the initial state vectors, the NASA JPL HORIZONS On-Line Ephemeris System \cite{JPL:HORIZONS} was queried with the following input settings:

{
\begin{itemize}
    \item \textit{Ephemeris Type:} Vector table
    \item \textit{Target Body:} Sun, Earth, Moon
    \item \textit{Coordinate Center:} Solar System Barycenter (SSB) [500@0]
    \item \textit{Time Specification:} Start = 2026-01-01 UTC, Stop = 2026-01-02, Step = 1 day
    \item \textit{Table Settings:} Defaults
\end{itemize}
}

Finally, in order to simplify computations, we will assume that the three celestial bodies which we consider, are perfectly spherical and of constant mass, therefore we will only consider the following values as constants for our celestial bodies:
{
\[
\begin{array}{lll}
    m_\text{Sun}=1988410\times 10^{24} \mathrm{kg}, & m_\text{Earth}=5.97219\times 10^{24} \mathrm{kg}, & m_\text{Moon}=7.349\times 10^{22} \mathrm{kg},\\
    R_\text{Sun}=695700 \mathrm{km}, & R_\text{Earth}=6371.01 \mathrm{km}, & R_\text{Moon}=1737.53 \mathrm{km}.
\end{array}
\]}

%%%%%%%%%%%%%%%%%%%%%%%%%%%%%%%%%%%%%%%%%%%%%%%%%%%%%%%%%%%%%%%
\subsection{Implementation}\label{subsec:implementation}

In order approximate the solutions to our three-body problem, i.e. to approximate the trajectories of our three celestial bodies, we will be using computer software, namely, all subsequent computations are carried out using the computer algebra system Maple. Numerical quantities are evaluated in high precision arithmetic with \texttt{Digits := 10}, and all final reported results will be given in scientific notation with the mantissa rounded to five decimal places. 

We start by writing the system of differential equations \eqref{eq:motion_equations_3_body_problem}, setting the universal gravitational constant \eqref{eq:valoareG} and utilizing the initial conditions obtained in Section \ref{subsec:data_gathering}; and we obtain an approximation of the solution of this Cauchy problem using the command \texttt{dsolve} with the option \texttt{numeric} (so that it makes use of the Runge-Kutta-Fehlberg method (RKF45)) in about $0.03-0.04\mathrm{s}$ on a mid-tier computer from 2020.

With our approximation of the solution we can start plotting the simulated trajectories of the Sun-Earth-Moon System with a duration of 1 year, from the 1st of January 2026 (00:00 UTC) to the 31st of December 2026 (23:59 UTC). In Figure \ref{fig:2D_trajectories}, we can visualize their projections into the $Oxy$ plane (as their $z$-component varies insignificantly). 

In Figure \ref{fig:2D_earth_moon_sun} we have all simulated trajectories combined, namely, in blue we have Earth's trajectory, in red we have the Moon's trajectory and the central orange dot of the figure represents the trajectory of the Sun. 

Note that the Sun also has movement, although at the scale of this figure it is not obvious. Changing the scale, in Figure \ref{fig:2D_sun} this motion becomes visible. Finally, in Figure \ref{fig:2D_earth_moon}, one can visualize how the Moon and Earth exchange positions relative to the Sun.

\begin{figure}[htbp]
    \centering
    \begin{subfigure}[b]{0.32\textwidth}
        \centering
        \includegraphics[width=\textwidth]{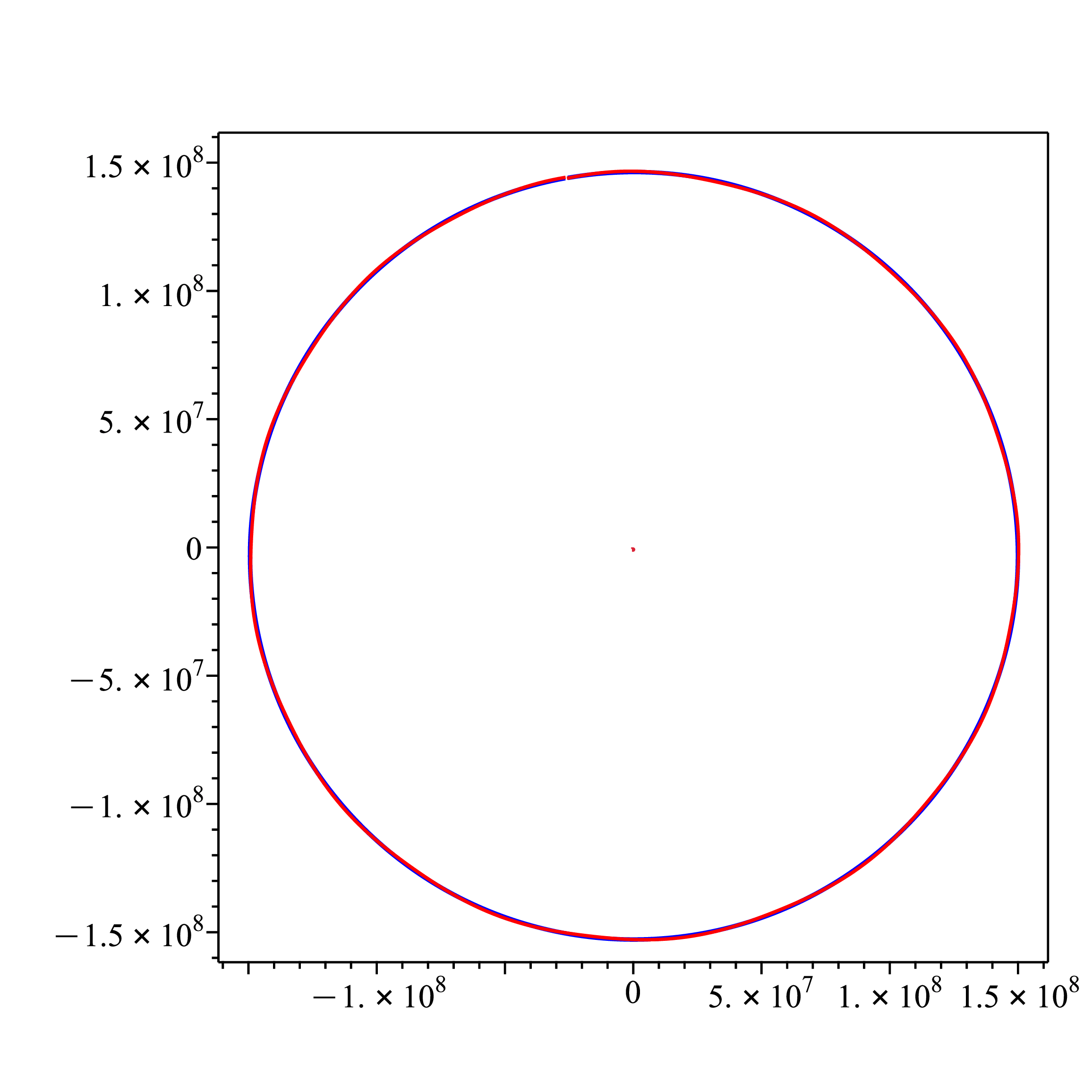}
        \caption{Sun, Earth and Moon}
        \label{fig:2D_earth_moon_sun}
    \end{subfigure}
    \begin{subfigure}[b]{0.32\textwidth}
        \centering
        \includegraphics[width=\textwidth]{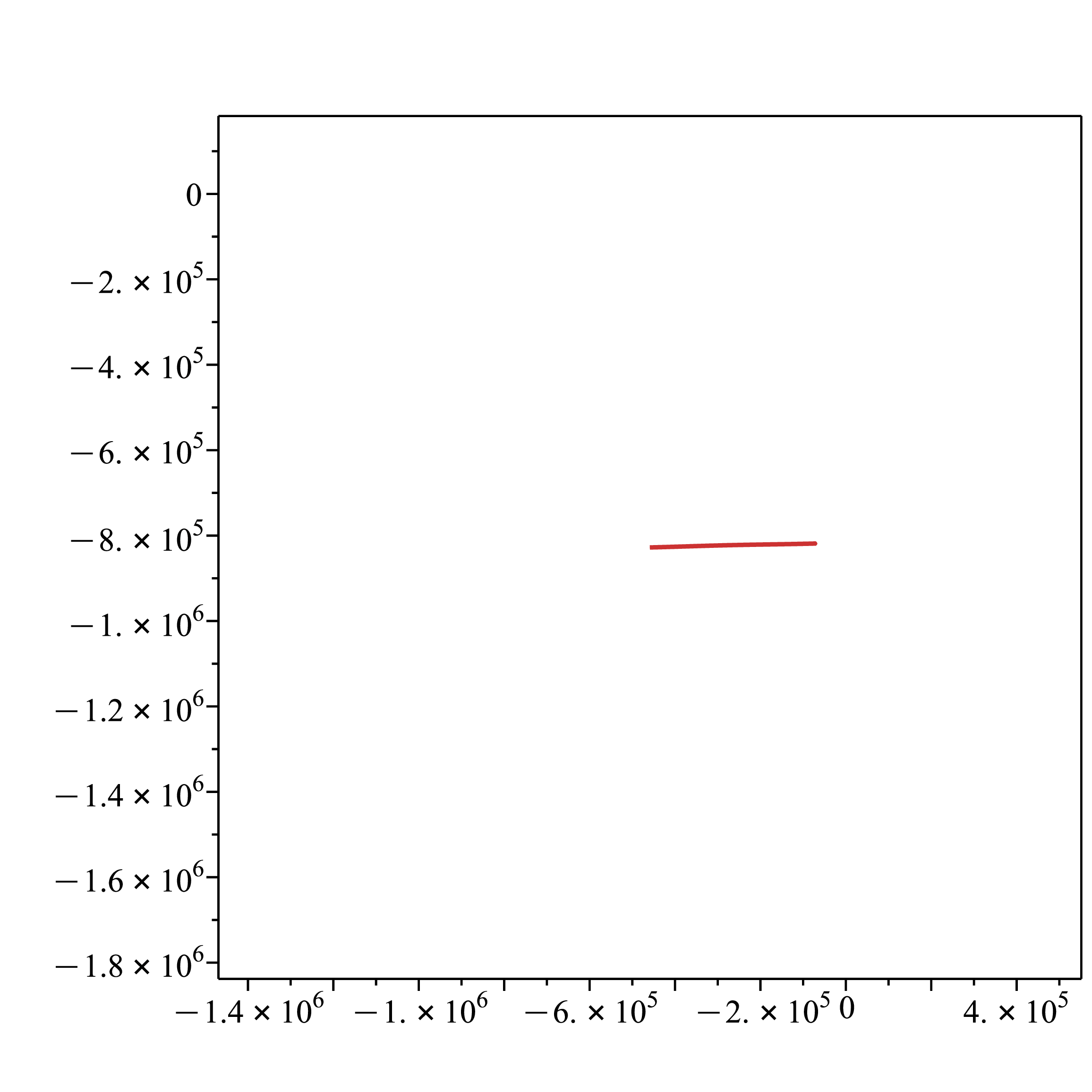}
        \caption{Sun}
        \label{fig:2D_sun}
    \end{subfigure}
    \begin{subfigure}[b]{0.32\textwidth}
        \centering
        \includegraphics[width=\textwidth]{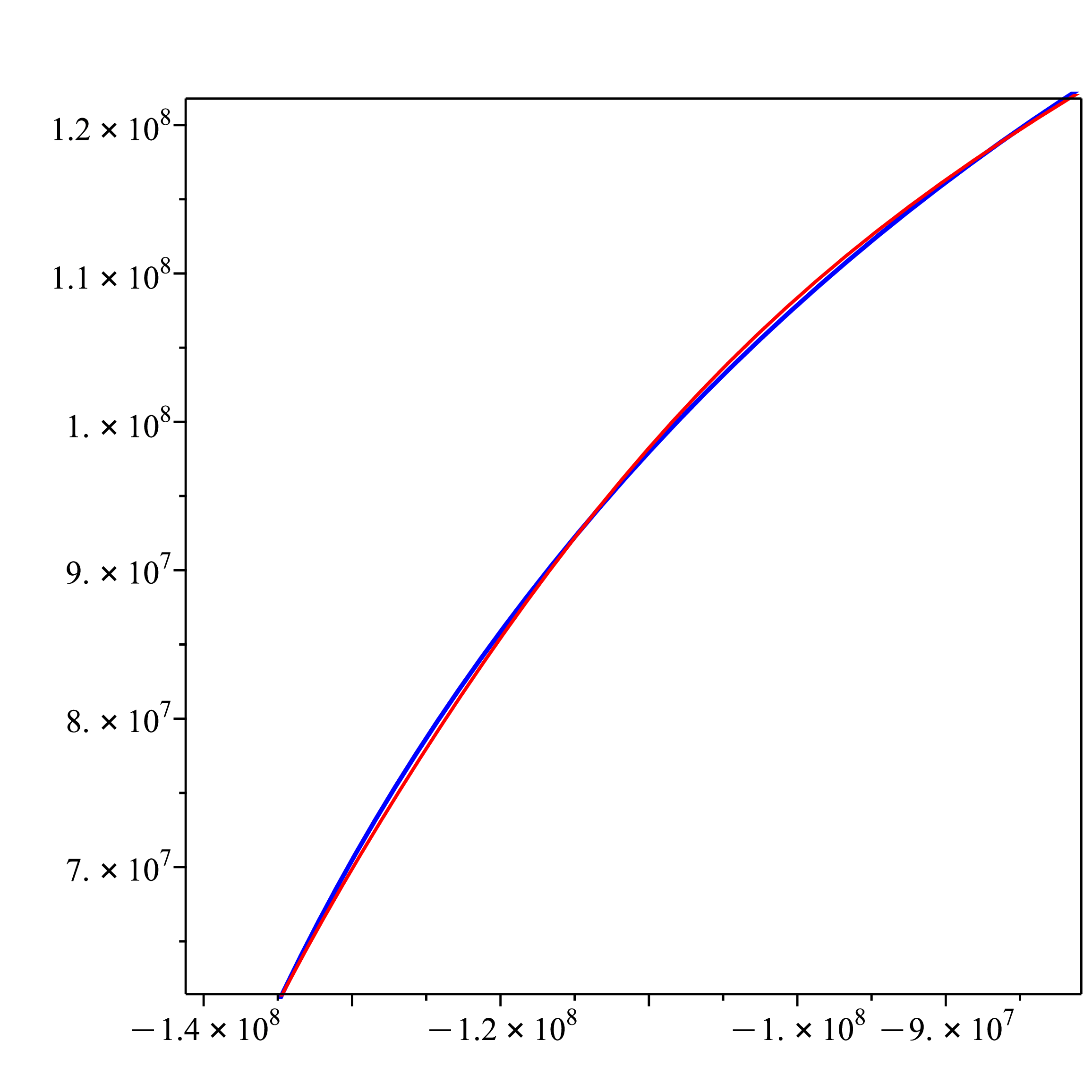}
        \caption{Earth and Moon}
        \label{fig:2D_earth_moon}
    \end{subfigure}
    \caption{Simulated 2D trajectories of the Sun-Earth-Moon System}
    \label{fig:2D_trajectories}
\end{figure}

On the same time scale, in Figure \ref{fig:3D_trajectories},  we can also visualize 3D renderings of the simulated trajectories of our three celestial bodies. Note that given the scale of the $Ox$ and $Oy$ axis, in order for the trajectories to not appear flat, we have changed the scale of the $Oz$-axis, and henceforth their motions are clearly happening in 3 dimensions. This being said, it is apparent that in our Cartesian frame, with coordinates provided by the NASA JPL HORIZONS On-Line Ephemeris System, the whole three-body system appears to be moving down in direction of the $Oz$-axis.

\begin{figure}[htbp]
    \centering
    \begin{subfigure}[b]{0.32\textwidth}
        \centering
        \includegraphics[width=\textwidth]{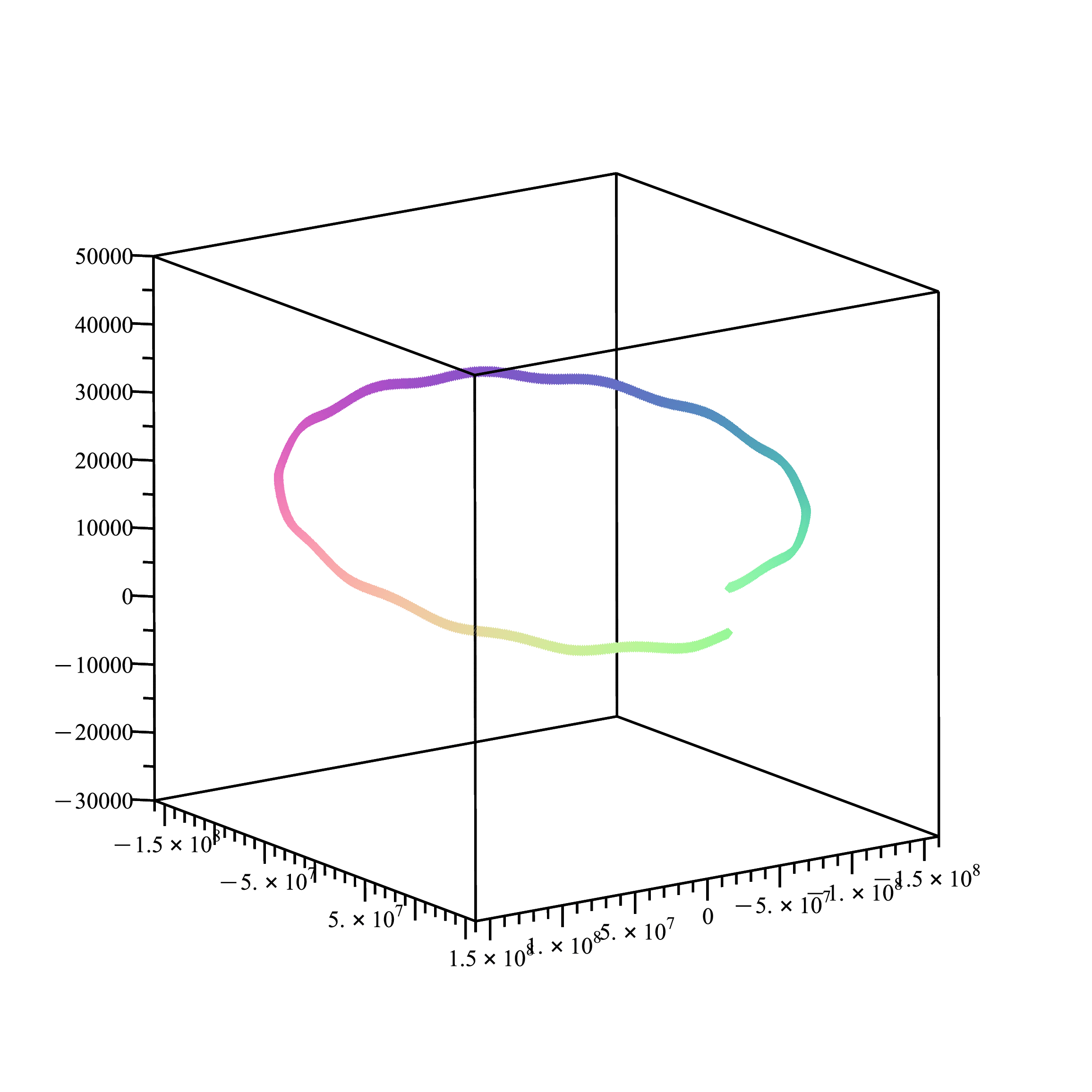}
        \caption{Earth}
        \label{fig:3D_earth}
    \end{subfigure}
    \begin{subfigure}[b]{0.32\textwidth}
        \centering
        \includegraphics[width=\textwidth]{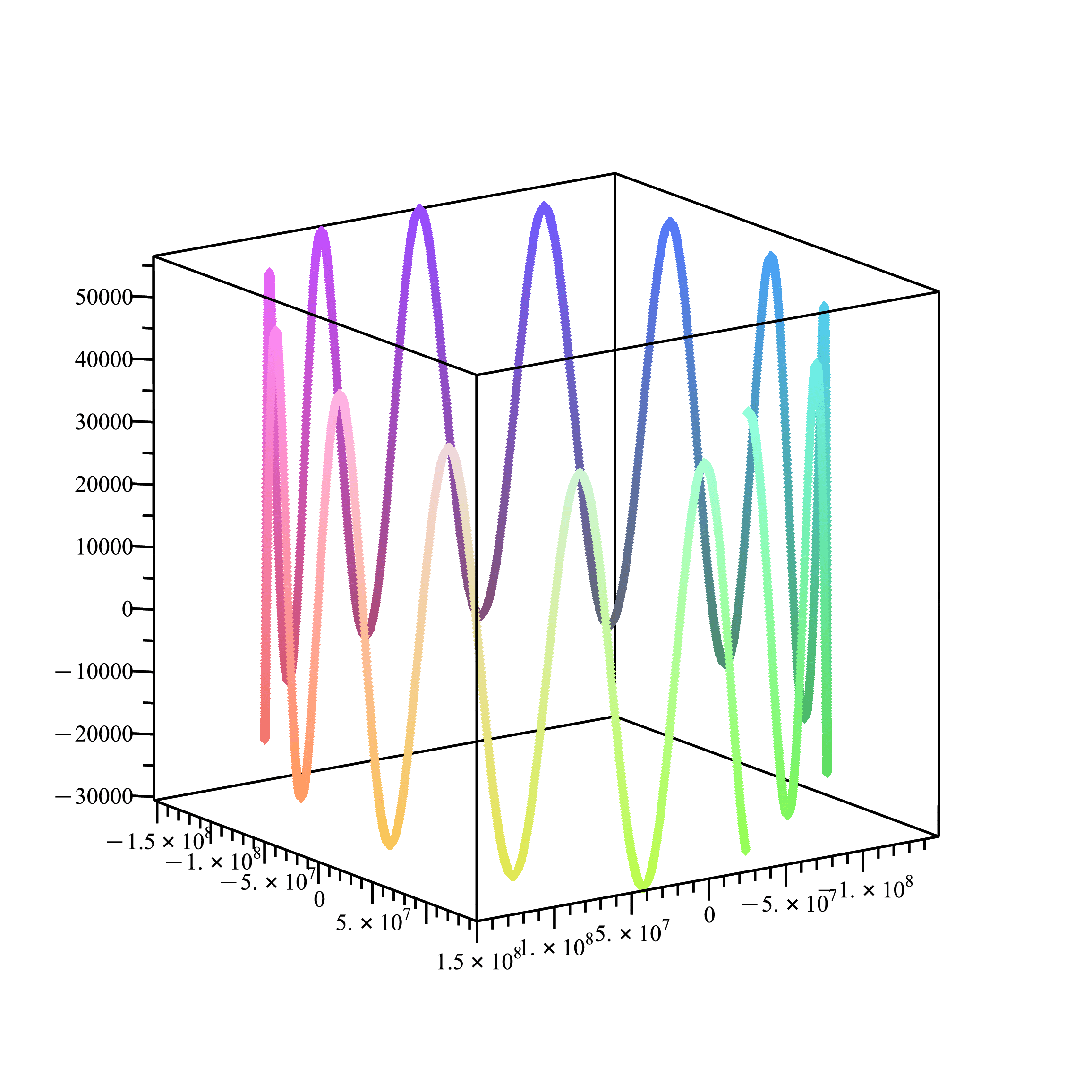}
        \caption{Moon}
        \label{fig:3D_moon}
    \end{subfigure}
    \begin{subfigure}[b]{0.32\textwidth}
        \centering
        \includegraphics[width=\textwidth]{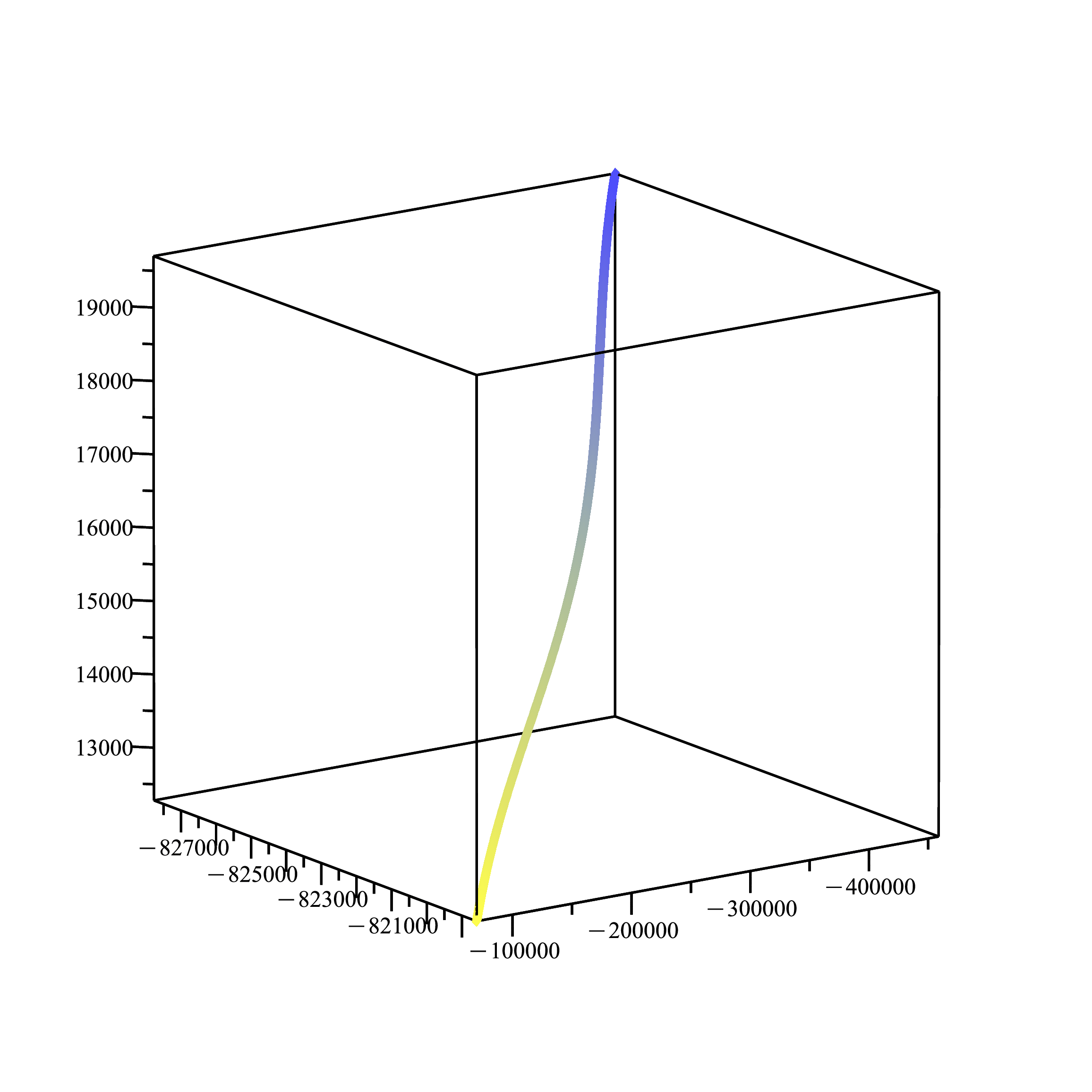}
        \caption{Sun}
        \label{fig:3D_sun}
    \end{subfigure}
    \caption{Simulated 3D trajectories of the Sun-Earth-Moon System}
    \label{fig:3D_trajectories}
\end{figure}

%%%%%%%%%%%%%%%%%%%%%%%%%%%%%%%%%%%%%%%%%%%%%%%%%%%%%%%%%%%%%%%
\subsection{Error analysis}\label{subsec:error}

In order to perform an error analysis, we have retrieved the estimated positions of the three celestial bodies in question, by querying the NASA JPL HORIZONS On-Line Ephemeris System \cite{JPL:HORIZONS}, with the same settings as discussed in Section \ref{subsec:data_gathering}, but for 1 year, 2 years, 5 years and respectively 10 years after our initial epoch $t_0 = \text{1 January 2026, 00:00 UTC}$, namely for the 1st of January 2027 (00:00 UTC), 1st of January 2028 (00:00 UTC), 1st of January 2031 (00:00 UTC) and 1st of January 2036 (00:00 UTC).

The values retrieved and the simulated positions of the three celestial bodies, together with the absolute errors of approximation (in both kilometers and astronomical units) are displayed in Table \ref{table:Earth_trajectory_errors}, Table \ref{table:Moon_trajectory_errors} and Table \ref{table:Sun_trajectory_errors} for Earth, the Moon and the Sun respectively. 

As we can see, the absolute errors keep increasing the further we get from the initial epoch $t_0$. This has multiple reasons:
\begin{enumerate}
    \item floating-point arithmetic error;
    \item approximation errors form data gathered in Section \ref{subsec:data_gathering}; 
    \item the estimation error of the RKF45 method;
    \item we do not take into account the other celestial bodies of the solar system;
    \item we do not take into account the general theory of relativity.
\end{enumerate}
While reason (4) is easily solvable by extending to the $n$-body problem and therefore taking into account all the known ``big'' celestial bodies for which we know initial positions and initial velocities at the common epoch $t_0$ (for which one may still easily use \cite{JPL:HORIZONS}), the other reasons would necessitate a complete change of the approach utilized. 

Despite these facts, for short-term simulations, in the range of $\pm10$ years, the discussed methodology appears to be sufficient to realize the goals of the current article, i.e. to predict solar eclipses on a short-time scale.

\begin{table}[H]
    \centering
    \renewcommand{\arraystretch}{1.3}
    \resizebox{0.8\textwidth}{!}{
    \begin{tabular}{||c|l|c|c||}
        \hline
        \hline
        \multirow[c]{2}{*}{Date (UTC)} & \multicolumn{1}{c|}{\multirow[c]{2}{*}{Coordinates (km)}} & \multicolumn{2}{c||}{Absolute error} \\
        \cline{3-4}
         &  & km & AU \\
        \hline
        \hline
        \multirow[c]{2}{*}{1 January 2026}
            & JPL: $(-2.65330\cdot 10^7,1.43946\cdot 10^8,1.08068\cdot 10^4)$ & \multirow[c]{2}{*}{$0$} & \multirow[c]{2}{*}{$0$}\\
        & Sim: $(-2.65330\cdot 10^7,1.43946\cdot 10^8,1.08068\cdot 10^4)$ & & \\
        \hline
        \multirow[c]{2}{*}{1 January 2027}
            & JPL: $(-2.55138\cdot 10^7,1.44166\cdot 10^8,3.85677\cdot 10^3)$ & \multirow[c]{2}{*}{$9.45737\cdot10^4$} & \multirow[c]{2}{*}{$0.00063$}\\
        & Sim: $(-2.54748\cdot 10^7,1.44080\cdot 10^8,4.14889\cdot 10^3)$ & & \\
        \hline
        \multirow[c]{2}{*}{1 January 2028}
            & JPL: $(-2.46334\cdot 10^7,1.44516\cdot 10^8,-2.27267\cdot 10^3)$ & \multirow[c]{2}{*}{$3.70591\cdot10^5$} & \multirow[c]{2}{*}{$0.00248$}\\
        & Sim: $(-2.44304\cdot 10^7,1.44206\cdot 10^8,-3.84546\cdot 10^3)$ & & \\
        \hline
        \multirow[c]{2}{*}{1 January 2031}
            & JPL: $(-2.54965\cdot 10^7,1.45092\cdot 10^8,-6.62711\cdot 10^2)$ & \multirow[c]{2}{*}{$1.90887\cdot10^6$} & \multirow[c]{2}{*}{$0.01276$}\\
        & Sim: $(-2.38495\cdot 10^7,1.44127\cdot 10^8,-2.62401\cdot 10^4)$ & & \\
        \hline
        \multirow[c]{2}{*}{1 January 2036}
            & JPL: $(-2.49862\cdot 10^7,1.43999\cdot 10^8,4.39021\cdot 10^2)$ & \multirow[c]{2}{*}{$3.83721\cdot10^6$} & \multirow[c]{2}{*}{$0.02565$}\\
        & Sim: $(-2.11618\cdot 10^7,1.44306\cdot 10^8,-6.30903\cdot 10^4)$ & & \\
        \hline
        \hline
    \end{tabular}}
    \caption{Absolute errors of the simulated trajectory of Earth for 10 years}
    \label{table:Earth_trajectory_errors}
\end{table}

\begin{table}[H]
    \centering
    \renewcommand{\arraystretch}{1.3}
    \resizebox{0.8\textwidth}{!}{
    \begin{tabular}{||c|l|c|c||}
        \hline
        \hline
        \multirow[c]{2}{*}{Date (UTC)} & \multicolumn{1}{c|}{\multirow[c]{2}{*}{Coordinates (km)}} & \multicolumn{2}{c||}{Absolute error} \\
        \cline{3-4}
         &  & km & AU \\
        \hline
        \hline
        \multirow[c]{2}{*}{1 January 2026}
            & JPL: $(-2.63888\cdot 10^7,1.44276\cdot 10^8,4.25602\cdot 10^4)$ & \multirow[c]{2}{*}{$0$} & \multirow[c]{2}{*}{$0$}\\
        & Sim: $(-2.63888\cdot 10^7,1.44276\cdot 10^8,4.25602\cdot 10^4)$ & & \\
        \hline
        \multirow[c]{2}{*}{1 January 2027}
            & JPL: $(-2.58697\cdot 10^7,1.44006\cdot 10^8,-2.76344\cdot 10^3)$ & \multirow[c]{2}{*}{$9.62038\cdot10^4$} & \multirow[c]{2}{*}{$0.00064$}\\
        & Sim: $(-2.58300\cdot 10^7,1.43918\cdot 10^8,-2.74075\cdot 10^4)$ & & \\
        \hline
        \multirow[c]{2}{*}{1 January 2028}
            & JPL: $(-2.43004\cdot 10^7,1.44285\cdot 10^8,1.24790\cdot 10^4)$ & \multirow[c]{2}{*}{$3.68605\cdot10^5$} & \multirow[c]{2}{*}{$0.00246$}\\
        & Sim: $(-2.40945\cdot 10^7,1.43979\cdot 10^8,1.13243\cdot 10^4)$ & & \\
        \hline
        \multirow[c]{2}{*}{1 January 2031}
            & JPL: $(-2.51183\cdot 10^7,1.45193\cdot 10^8,2.73435\cdot 10^4)$ & \multirow[c]{2}{*}{$1.88957\cdot10^6$} & \multirow[c]{2}{*}{$0.01263$}\\
        & Sim: $(-2.34779\cdot 10^7,1.44255\cdot 10^8,2.30141\cdot 10^2)$ & & \\
        \hline
        \multirow[c]{2}{*}{1 January 2036}
            & JPL: $(-2.47507\cdot 10^7,1.43668\cdot 10^8,1.35308\cdot 10^4)$ & \multirow[c]{2}{*}{$3.92144\cdot10^6$} & \multirow[c]{2}{*}{$0.02621$}\\
        & Sim: $(-2.08486\cdot 10^7,1.44051\cdot 10^8,-5.94879\cdot 10^4)$ & & \\
        \hline
        \hline
    \end{tabular}}
    \caption{Absolute errors of the simulated trajectory of the Moon for 10 years}
    \label{table:Moon_trajectory_errors}
\end{table}

\begin{table}[H]
    \centering
    \renewcommand{\arraystretch}{1.3}
    \resizebox{0.8\textwidth}{!}{
    \begin{tabular}{||c|l|c|c||}
        \hline
        \hline
        \multirow[c]{2}{*}{Date (UTC)} & \multicolumn{1}{c|}{\multirow[c]{2}{*}{Coordinates (km)}} & \multicolumn{2}{c||}{Absolute error} \\
        \cline{3-4}
         &  & km & AU \\
        \hline
        \hline
        \multirow[c]{2}{*}{1 January 2026}
            & JPL: $(-4.58863\cdot 10^5,-8.27774\cdot 10^5,1.96997\cdot 10^4)$ & \multirow[c]{2}{*}{$0$} & \multirow[c]{2}{*}{$0$}\\
        & Sim: $(-4.58863\cdot 10^5,-8.27774\cdot 10^5,1.96997\cdot 10^4)$ & & \\
        \hline
        \multirow[c]{2}{*}{1 January 2027}
            & JPL: $(-1.05636\cdot 10^5,-7.28935\cdot 10^5,1.24094\cdot 10^4)$ & \multirow[c]{2}{*}{$9.65127\cdot10^4$} & \multirow[c]{2}{*}{$0.00065$}\\
        & Sim: $(-6.98860\cdot 10^4,-8.18582\cdot 10^5,1.22799\cdot 10^4)$ & & \\
        \hline
        \multirow[c]{2}{*}{1 January 2028}
            & JPL: $(1.21374\cdot 10^5,-4.94673\cdot 10^5,7.04701\cdot 10^3)$ & \multirow[c]{2}{*}{$3.71677\cdot10^5$} & \multirow[c]{2}{*}{$0.00248$}\\
        & Sim: $(3.19091\cdot 10^5,-8.09390\cdot 10^5,4.86015\cdot 10^3)$ & & \\
        \hline
        \multirow[c]{2}{*}{1 January 2031}
            & JPL: $(-1.46526\cdot 10^5,1.91327\cdot 10^5,9.68541\cdot 10^3)$ & \multirow[c]{2}{*}{$1.90168\cdot10^6$} & \multirow[c]{2}{*}{$0.01271$}\\
        & Sim: $(1.48710\cdot 10^6,-7.81786\cdot 10^5,-1.74195\cdot 10^4)$ & & \\
        \hline
        \multirow[c]{2}{*}{1 January 2036}
            & JPL: $(-3.71614\cdot 10^5,-1.03700\cdot 10^6,1.25488\cdot 10^4)$ & \multirow[c]{2}{*}{$3.81716\cdot10^6$} & \multirow[c]{2}{*}{$0.02552$}\\
        & Sim: $(3.43305\cdot 10^6,-7.35798\cdot 10^5,-5.45386\cdot 10^4)$ & & \\
        \hline
        \hline
    \end{tabular}}
    \caption{Absolute errors of the simulated trajectory of the Sun for 10 years}
    \label{table:Sun_trajectory_errors}
\end{table}

%%%%%%%%%%%%%%%%%%%%%%%%%%%%%%%%%%%%%%%%%%%%%%%%%%%%%%%%%%%%%%%
%%%%%%%%%%%%%%%%%%%%%%%%%%%%%%%%%%%%%%%%%%%%%%%%%%%%%%%%%%%%%%%
%%%%%%%%%%%%%%%%%%%%%%%%%%%%%%%%%%%%%%%%%%%%%%%%%%%%%%%%%%%%%%%
\section{Predicting solar eclipses} \label{sec:results}

%%%%%%%%%%%%%%%%%%%%%%%%%%%%%%%%%%%%%%%%%%%%%%%%%%%%%%%%%%%%%%%
\subsection{Eclipse detection algorithm}
\label{subsec:eclipse_detection}

As in Section \ref{subsec:the_geometry_of_eclipses}, we consider a luminous spherical body $\Sun$ of center $S(x_S,y_S,z_S)$ and radius $R_S$; and an opaque spherical body $\M$ of center $M(x_M,y_M,z_M)$ and radius $R_M$ together with the assumption that $R_S>R_M$. Additionally, now, we also consider another spherical body $\mathcal{E}$ of center $E(x_E,y_E,z_E)$ and radius $R_E$.

In this section we want to obtain a criterion to determine  when does an observer residing on $\mathcal{E}$ can experience an eclipse of $\Sun$ by $\M$.

\begin{algorithm} \label{algo:main}
    \textbf{Step 1.}
    Given the parameters of three spherical bodies $\Sun$, $\M$ and $\mathcal{E}$, as above, via Remark \ref{remark:eclipse_parameters}, we now know the parameters of the three eclipse cones: the umbral cone (i.e., the circular cone $\U$), the antumbral cone (i.e., the convex circular cone $\A$) and the penumbral cone (i.e., the convex circular cone frustum $\P$).

    \medskip
    \textbf{Step 2.} We apply Algorithm \ref{algo:intersection_sphere_double_cone} to detect if there is an intersection between the sphere $\mathcal{E}$ and the umbral double cone $\mathcal{DU}$.
    \begin{itemize}[leftmargin=*]
        \item If there is a nonempty intersection between the two, we further refine by applying Algorithm \ref{algo:intersection_sphere_cone} to detect if $\mathcal{E}$ intersects the umbral cone $\U$, and Algorithm \ref{algo:intersection_sphere_convex_cone} to detect if $\mathcal{E}$ intersects the antumbral cone $\A$. If the intersection with either one of $\U$ and $\A$ is nonempty, then an eclipse of $\Sun$ by $\M$ can be observed from $\mathcal{E}$. In this case, by Remark \ref{remark:total_vs_annular_eclipse} we may further classify it into being either a total eclipse, either an annular eclipse. If both intersections are empty, then we move on to the next scenario (\textbf{Step 3}).
        \item If $\mathcal{E}$ and $\mathcal{DU}$ are disjoint we also move onto \textbf{Step 3}.
    \end{itemize}
    
    \medskip
    \textbf{Step 3.} We apply Algorithm \ref{algo:intersection_sphere_double_cone} to detect if there is an intersection between the sphere $\mathcal{E}$ and the penumbral double cone $\mathcal{DP}$.
    \begin{itemize}[leftmargin=*]
        \item If there is a nonempty intersection between the two, we apply Algorithm \ref{algo:intersection_sphere_frustum} to detect if $\mathcal{E}$ intersects the penumbral cone $\P$. If the intersection with $\P$ is nonempty, then a partial eclipse of $\Sun$ by $\M$ can be observed from $\mathcal{E}$. Otherwise, if $\mathcal{E}$ and $\P$ are disjoint, then there is no eclipse of $\Sun$ by $\M$ observable from $\mathcal{E}$.
        \item If $\mathcal{E}$ and $\mathcal{DP}$ are disjoint, then again there is no eclipse of $\Sun$ by $\M$ observable from $\mathcal{E}$.
    \end{itemize}
\end{algorithm}

\begin{remark} \label{rem:multiple_eclipses}
    Note that multiple types of eclipses of $\Sun$ by $\M$ can be experienced at the same time by observers from $\mathcal{E}$, although from different locations on $\mathcal{E}$. This being said, in the current form of Algorithm \ref{algo:main}, total eclipses and annular eclipses take precedence to partial eclipses. In other words, if \textbf{Step 2} of Algorithm \ref{algo:main} detects a total eclipse, or an annular eclipse, even if somewhere on $\mathcal{E}$ a partial eclipse can be experienced, \textbf{Step 3} should not be verified as it can give a false-positive because of the total/annular eclipse. The reason for this lies in the fact that the penumbral cone includes in its interior both the umbral cone and the antumbral one. 
 \end{remark}

%%%%%%%%%%%%%%%%%%%%%%%%%%%%%%%%%%%%%%%%%%%%%%%%%%%%%%%%%%%%%%%
\subsection{Implementation and error analysis} \label{subsec:implentation_and_errors}

Using the computer algebra system Maple (with the same configuration as in Section \ref{subsec:implementation}), we apply Algorithm \ref{algo:main} on the simulated positions of the three celestial bodies of the Sun-Earth-Moon System obtained in Section \ref{subsec:implementation}, over a duration of 1 year, from the 1st of January 2026 (00:00 UTC) to the 31st of December 2026  (23:59 UTC), checking the algorithm minute by minute.

In Table \ref{table:17_febr_2026} and Table  \ref{table:12_aug_2026}, we report the detected solar eclipses, their type (i.e., total eclipse, annular eclipse or partial eclipse), their starting and finishing time, together with the  values reported by NASA \cite[pp. 230--231]{nasa} and an estimation of the absolute errors of approximation. Note beforehand that the errors seen are consequences of the errors of the trajectory simulations (see Section \ref{subsec:error}) and not of Algorithm \ref{algo:main}, as the proposed algorithm, being mostly geometric in nature (see Section \ref{sec:geometry}), has negligible errors. 

\begin{table}[H]
    \centering
    \renewcommand{\arraystretch}{1.3}
    \resizebox{1\textwidth}{!}{
    \begin{tabular}{||l||c|c|c|c||}
        \hline
        \hline
        & Start of Partial Eclipse & Start of Annular Eclipse & End of Annular Eclipse & End of Partial Eclipse\\
        \hline
        \hline
        Simulated time (UTC)&09:56&11:41&12:40&14:26\\
        \hline
        NASA time (UTC)&09:56&11:42&12:41&14:27\\
        \hline
        Absolute error (min.)&0&1&1&1\\
        \hline
        \hline
    \end{tabular}}
    \caption{The annular solar eclipse of 17th February 2026}
    \label{table:17_febr_2026}
\end{table}

\begin{table}[H]
    \centering
    \renewcommand{\arraystretch}{1.3}
    \resizebox{1\textwidth}{!}{
    \begin{tabular}{||l||c|c|c|c||}
        \hline
        \hline
        & Start of Partial Eclipse & Start of Total Eclipse & End of Total Eclipse & End of Partial Eclipse\\
        \hline
        \hline
        Simulated time (UTC)&15:22&16:45&18:21&19:44\\
        \hline
        NASA time (UTC)&15:33&16:57&18:33&19:57\\
        \hline
        Absolute error (min.)&11&12&12&13\\
        \hline
        \hline
    \end{tabular}}
    \caption{The total solar eclipse of 12th  August 2026}
    \label{table:12_aug_2026}
\end{table}

%%%%%%%%%%%%%%%%%%%%%%%%%%%%%%%%%%%%%%%%%%%%%%%%%%%%%%%%%%%%%%%
%%%%%%%%%%%%%%%%%%%%%%%References%%%%%%%%%%%%%%%%%%%%%%%%%%%%%%
%%%%%%%%%%%%%%%%%%%%%%%%%%%%%%%%%%%%%%%%%%%%%%%%%%%%%%%%%%%%%%%
\phantomsection

\end{document}